\documentclass[10pt,twocolumn]{IEEEtran}
\usepackage{amsmath,amssymb,cite,multirow}
\usepackage{graphicx}
\usepackage{url}

\newtheorem{proposition}{Proposition}

\newtheorem{lemma}{Lemma}
\newtheorem{corollary}{Corollary}

\newcommand{\df}{\stackrel{\mbox{\scriptsize def}}{=}}

\newcommand{\gf}{\mathrm{GF}}

\newcommand{\As}{A_{\mbox{\tiny{S}}}}
\newcommand{\Am}{A_{\mbox{\tiny{M}}}}
\newcommand{\Ac}{A_{\mbox{\tiny{C}}}}
\newcommand{\Ai}{A_{\mbox{\tiny{I}}}}

\newcommand{\as}{a_{\mbox{\tiny{S}}}}

\newcommand{\ai}{a_{\mbox{\tiny{I}}}}

\newcommand{\ds}{d_{\mbox{\tiny{S}}}}

\newcommand{\di}{d_{\mbox{\tiny{I}}}}

\newcommand{\Ks}{K_{\mbox{\tiny{S}}}}

\newcommand{\Kc}{K_{\mbox{\tiny{C}}}}
\newcommand{\Ki}{K_{\mbox{\tiny{I}}}}

\newcommand{\ks}{k_{\mbox{\tiny{S}}}}

\newcommand{\ki}{k_{\mbox{\tiny{I}}}}

\newcommand{\Ns}{N_{\mbox{\tiny{S}}}}

\newcommand{\Ni}{N_{\mbox{\tiny{I}}}}

\newcommand{\Vs}{V_{\mbox{\tiny{S}}}}

\newcommand{\Vi}{V_{\mbox{\tiny{I}}}}

\begin{document}
\title{Packing and Covering Properties of Subspace Codes for Error Control in Random Linear Network Coding}
\author{Maximilien Gadouleau,~\IEEEmembership{Member,~IEEE,}
and Zhiyuan Yan,~\IEEEmembership{Senior Member,~IEEE}%
\thanks{This work was supported in part by Thales Communications
Inc. and in part by a grant from the Commonwealth of Pennsylvania,
Department of Community and Economic Development, through the
Pennsylvania Infrastructure Technology Alliance (PITA). The work of Zhiyuan Yan was supported in part by a
summer extension grant from Air Force Research Laboratory, Rome, New
York, and the work of Maximilien Gadouleau was supported in part by the ANR project RISC.  Part of the
material in this paper was presented at 2009 IEEE International Symposium on Information Theory (ISIT 2009) in Seoul, Korea.} %
\thanks{Maximilien Gadouleau was with the Department of Electrical and Computer Engineering, Lehigh University, Bethlehem, PA 18015 USA. Now he is with CReSTIC, Universit\'e de Reims Champagne-Ardenne, Reims 51100 France. Zhiyuan Yan is with the Department of Electrical and Computer Engineering, Lehigh University, Bethlehem, PA, 18015 USA (e-mail: maximilien.gadouleau@univ-reims.fr; yan@lehigh.edu).}}

\maketitle

\thispagestyle{empty}

\begin{abstract}
Codes in the projective space and codes in the Grassmannian over a
finite field --- referred to as subspace codes and
constant-dimension codes (CDCs), respectively --- have been proposed for error control in random linear network coding. For subspace codes
and CDCs, a subspace metric was introduced to correct both errors
and erasures, and an injection metric was proposed to correct
adversarial errors. In this paper, we investigate the
packing and covering properties of subspace codes with both metrics.
We first determine some fundamental geometric properties of the
projective space with both metrics. Using these properties, we then
derive bounds on the cardinalities of packing and covering subspace
codes, and determine the asymptotic rates of optimal packing and
optimal covering subspace codes with both metrics. Our results not
only provide guiding principles for the
code design for error control in random linear network coding, but also illustrate the difference between the two metrics from a
geometric perspective. In
particular, our results show that optimal packing CDCs are optimal
packing subspace codes up to a scalar for both metrics if and only
if their dimension is half of their length (up to rounding). In this
case, CDCs suffer from only limited rate loss as opposed to subspace
codes with the same minimum distance. We also show that optimal
covering CDCs can be used to construct asymptotically optimal
covering subspace codes with the injection metric only.
\end{abstract}


\section{Introduction}\label{sec:introduction}
Due to its vector-space preserving property, random linear network coding
\cite{ho_it06, HL08} can be viewed as transmitting subspaces
over an operator channel \cite{koetter_it08}. As such, error control for
random linear network coding can be modeled as a coding
problem, where codewords are subspaces and the distance is measured by either
the subspace distance \cite{koetter_it08} or
the injection metric \cite{silva_it09}. Codes in the
projective space, referred to as subspace codes henceforth, and codes in the Grassmannian, referred to as
constant-dimension codes (CDCs) henceforth, have been both investigated for error
control in random linear network coding. Using CDCs is sometimes advantageous since the fixed dimension of CDCs simplifies the network protocol somewhat \cite{koetter_it08}.

The construction and properties of CDCs thus have attracted a lot of
attention. Different constructions of CDCs have been proposed
\cite{koetter_it08,silva_it08,skachek_arxiv08,gadouleau_it09_cdc}.
Bounds on CDCs based on packing properties are investigated (see,
for example, \cite{koetter_it08,skachek_arxiv08, kohnert_mmics08,
xia_dcc09}), and the covering properties of CDCs are investigated in
\cite{gadouleau_it09_cdc}. The construction and properties of
subspace codes have received less consideration, and previous
works on  subspace codes (see, for example, \cite{etzion_isit08,etzion_it09,gabidulin_isit08}) have focused on the packing properties. In \cite{etzion_isit08}, bounds on the maximum cardinality of a
subspace code with the subspace metric, notably the counterpart of
the Gilbert bound, are derived. Another bound relating the maximum cardinality of CDCs to that of subspace codes is given in \cite{etzion_it09}. Bounds and constructions of subspace
codes are also investigated in \cite{gabidulin_isit08}. Despite the
previous works, two significant problems remain  open. First, despite the
aforementioned advantage of CDCs, what is the rate loss of CDCs as
opposed to subspace codes of the same minimum distance and hence
error correction capability? Since random linear network coding achieves multicast capacity with probability exponentially approaching 1 with the length of the code \cite{ho_it06}, the asymptotic rates of subspace codes and asymptotic rate loss of CDCs are both significant. The second problem involves the two
metrics that have been introduced for subspace codes: what is the
difference between the two metrics proposed for subspace codes and
CDCs beyond those discussed in \cite{silva_it09}? Note that the
two questions are somewhat related, since the first question is
applicable for both metrics. The answers to these questions are
significant to the code design for error control in random linear network coding.

Aiming to answer these two questions, our work in this paper focuses on the packing and covering properties of subspace codes. Packing and covering properties not only are interesting in their own right as fundamental geometric properties, also are significant for various practical purposes. First, our work is motivated by their significance to design and decoding of subspace codes. Since a code can be viewed as a packing of its ambient space, the significance of packing properties is clear. In contrast, the importance of covering properties is more subtle and deserves more explanation. For example, a class of nearly optimal CDCs, referred to as liftings of rank metric codes, have covering radii no less than their minimum distance and thus are not optimal CDCs \cite{gadouleau_it09_cdc}. This example shows how a covering property is relevant to the design of subspace codes. The covering radius also characterizes the decoding performance of a code, since it is the maximum weight of a decodable error by minimum distance decoding \cite{pless_book98} and also has applications to decoding with erasures \cite{cohen_book97}. Second, covering properties are also important for other reasons. For example,
covering properties are important for the security of keystreams against cryptanalytic attacks \cite{HK01}.

Our main contributions of this paper are that for both metrics, we first determine some
fundamental geometric properties of the projective space, and then
use these properties to derive bounds and to determine the
asymptotic rates of subspace codes based on packing and covering.
Our results provide some answers to both open problems above. First, our
results show that for both metrics optimal packing CDCs are optimal
packing subspace codes up to a scalar if and only if their dimension
is half of their length (up to rounding), which implies that in this
case CDCs suffer from a limited rate loss as opposed to subspace
codes with the same minimum distance.  Furthermore, when the asymptotic
rate of subspace codes is fixed, the relative subspace distance of optimal subspace codes is twice as much as the relative injection distance. Second, our
results illustrate the difference between the two metrics from a
geometric perspective. Above all, the projective space has different
geometric properties under the two metrics. The different geometric
properties further result in different asymptotic rates of covering
codes with the two metrics. With the injection metric, optimal
covering CDCs can be used to construct asymptotically optimal
covering subspace codes. However, with the subspace metric, this
does not hold.

To the best of our knowledge, our results on the geometric properties of the
projective space are novel, and our investigation of covering properties
of subspace codes is the first one in the literature. Note that our investigation of covering properties differs from the study in \cite{gadouleau_it09_cdc}: while how CDCs cover the Grassmannian
was investigated in \cite{gadouleau_it09_cdc}, we consider how subspace codes cover the whole
projective space in this paper. Our investigation of packing properties leads to tighter bounds than the Gilbert bound in \cite{etzion_isit08}, and our relation between the optimal cardinalities of subspace codes and CDCs is also more precise than that in \cite{etzion_it09}. Our asymptotic rates based on packing properties also appear to be novel.


The rest of the paper is organized as follows.
Section~\ref{sec:preliminaries} reviews necessary background on
subspace codes, CDCs, and related concepts. In
Section~\ref{sec:subspace}, we investigate the packing and covering
properties of subspace codes with the subspace metric. In
Section~\ref{sec:injection}, we study the packing and covering
properties of subspace codes with the injection metric. Finally, Section~\ref{sec:conclusion} summarizes our results and provides future work directions.

\section{Preliminaries}\label{sec:preliminaries}

We refer to the set of all subspaces of $\mathrm{GF}(q)^n$ with
dimension $r$ as the Grassmannian of dimension $r$ and denote it as
$E_r(q,n)$; we refer to $E(q,n) = \bigcup_{r=0}^n E_r(q,n)$ as the
projective space. We have $|E_r(q,n)| = {n \brack r}$, where ${n \brack r} =
\prod_{i=0}^{r-1} \frac{q^n - q^i}{q^r - q^i}$ is the Gaussian
binomial \cite{andrews_book76}. A very instrumental result
\cite{gadouleau_it08_dep} about the Gaussian binomial is that for
all $0 \leq r \leq n$:
\begin{equation}\label{eq:Gaussian}
    q^{r(n-r)} \leq {n \brack r} < K_q^{-1} q^{r(n-r)},
\end{equation}
where $K_q = \prod_{j=1}^\infty (1-q^{-j})$ represents the ratio of non-singular matrices in $\gf(q)^{n \times n}$ as $n$ tends to infinity. By definition,  $K_q = \phi(q^{-1})$, where $\phi$ is the Euler function. Furthermore, by the pentagonal number theorem, $K_q = \sum_{n=-\infty}^\infty (-1)^n q^{(n-3n^2)/2}$ \cite{gasper_book04}. Finally, we also have $K_q^{-1} = \sum_{k=0}^\infty p(k) q^{-k}$, where $p(k)$ is the partition number of $k$ \cite{andrews_book76}.

For $U,V \in E(q,n)$,
both the \emph{subspace metric} \cite[(3)]{koetter_it08} $\ds(U,V)
\df \dim(U + V) - \dim(U \cap V)$ and \emph{injection metric}
\cite[Def.~1]{silva_it09}
\begin{eqnarray}
    \nonumber
    \di(U,V) &\df& \frac{1}{2} \ds(U,V) + \frac{1}{2} |\dim(U) - \dim(V)|\\
    \label{eq:dm1}
    &=& \max\{\dim(U), \dim(V)\} - \dim(U \cap V)\\
    \label{eq:dm2}
    &=& \dim(U + V) - \min\{\dim(U), \dim(V)\}\\
    \nonumber
    &\geq& |\dim(U) - \dim(V)|
\end{eqnarray}
are metrics over $E(q,n)$. For all $U,V \in E(q,n)$,
\begin{equation}\label{eq:ds_dm}
    \frac{1}{2} \ds(U,V) \leq \di(U,V) \leq \ds(U,V),
\end{equation}
and $\di(U,V) = \frac{1}{2} \ds(U,V)$ if and only if $\dim(U) =
\dim(V)$, and $\di(U,V) = \ds(U,V)$ if and only if $U \subseteq V$
or $V \subseteq U$.

A {\em subspace code} is a nonempty subset of $E(q,n)$. The minimum
subspace (respectively, injection) distance of a subspace code is
the minimum subspace (respectively, injection) distance over all
pairs of distinct codewords. A subset of $E_r(q,n)$ is called a
constant-dimension code (CDC). A CDC is thus a subspace code whose
codewords have the same dimension. Since for CDCs $\di(U,V) = \frac{1}{2} \ds(U,V)$, we focus on the injection metric when considering CDCs. We denote the \textbf{maximum}
cardinality of a CDC in $E_r(q,n)$ with \textbf{minimum injection distance}
$d$ as $\Ac(q,n,r,d)$. We have $\Ac(q,n,r,d) = \Ac(q,n,n-r,d)$,
$\Ac(q,n,r,1) = {n \brack r}$  and it is shown
\cite{gadouleau_it09_cdc, xia_dcc09} for $r \leq \left\lfloor
\frac{n}{2}\right\rfloor$ and $2 \leq d \leq r$,
\begin{eqnarray}\nonumber
    q^{(n-r)(r-d+1)} + 1
    &\leq & \Ac(q,n,r,d)\\ &\leq &
    \nonumber
    \frac{{n \brack r-d+1}}{{r \brack r-d+1}} < K_q^{-1} q^{(n-r)(r-d+1)}.\\
    \label{eq:bounds_Ac}
\end{eqnarray}
The lower bound on $\Ac(q,n,r,d)$ in (\ref{eq:bounds_Ac}) is implicit from the code construction in \cite{gadouleau_it09_cdc}, and the upper bounds on $\Ac(q,n,r,d)$ in (\ref{eq:bounds_Ac}) are from \cite{koetter_it08}. Thus,
CDCs in $E_r(q,n)$ ($r \leq \left\lfloor \frac{n}{2} \right\rfloor$) with minimum injection distance $d$ and cardinality
$q^{(n-r)(r-d+1)}$ proposed in \cite{koetter_it08} are optimal up to a scalar; we refer to these CDCs as KK codes
henceforth. The covering
radius in $E_r(q,n)$ of a CDC $\mathcal{C}$ is defined as $\max_{U
\in E_r(q,n)} \di(U,\mathcal{C})$. We also denote the
\textbf{minimum} cardinality of a CDC with covering radius $\rho$ in
$E_r(q,n)$ as $\Kc(q,n,r,\rho)$ \cite{gadouleau_it09_cdc}. It was
shown in \cite{gadouleau_it09_cdc} that $\Kc(q,n,r,\rho)$ is on the
order of $q^{r(n-r) - \rho(n-\rho)}$, and an asymptotically optimal construction of covering CDCs is designed in \cite[Proposition 12]{gadouleau_it09_cdc}.

\section{Packing and covering properties of subspace
codes with the subspace metric}\label{sec:subspace}

\subsection{Properties of balls with subspace radii}\label{sec:balls_subspace}

We first investigate the properties of balls with subspace radii in
$E(q,n)$, which will be instrumental in our study of packing and
covering properties of subspace codes with the subspace metric. We
first derive bounds on $|E(q,n)|$ below.
In order to simplify notations, we denote $\theta(q) \df \sum_{n=0}^\infty q^{-n^2}$, which is related to the Jacobi theta function $\vartheta_3(z,q) = \sum_{n=-\infty}^\infty q^{n^2} e^{2niz}$ by $\theta(q) = \frac{1}{2}\left[\vartheta_3(0,q^{-1}) + 1\right]$ \cite{abramowitz_book65}. We remark that $\theta(q) > 1$ for all $q \geq 2$, and that $\theta(q)$ is a decreasing function of $q$ and approaches $1$ as $q$ tends to infinity.

\begin{lemma}\label{lemma:bounds_E}
For all $n$, $q^{\left\lfloor \frac{n}{2} \right\rfloor(n-\left\lfloor \frac{n}{2} \right\rfloor)} \leq |E(q,n)| < 2 K_q^{-1} \theta(q)
q^{\left\lfloor \frac{n}{2} \right\rfloor(n-\left\lfloor \frac{n}{2} \right\rfloor)}$.
\end{lemma}

\begin{proof}
We have $|E(q,n)| = \sum_{r=0}^n {n \brack r} \geq {n \brack \left\lfloor \frac{n}{2} \right\rfloor}
\geq q^{\left\lfloor \frac{n}{2} \right\rfloor(n-\left\lfloor \frac{n}{2} \right\rfloor)}$ by (\ref{eq:Gaussian}), which proves the lower
bound. Also, ${n \brack n-r} = {n \brack r}$ and hence $\sum_{r=0}^n
{n \brack r} \leq 2 \sum_{r=0}^{\left\lfloor \frac{n}{2} \right\rfloor} {n \brack r} < 2 K_q^{-1}
\sum_{r=0}^{\left\lfloor \frac{n}{2} \right\rfloor} q^{r(n-r)}$ by (\ref{eq:Gaussian}). Therefore,
$|E(q,n)| < 2 K_q^{-1} q^{\left\lfloor \frac{n}{2} \right\rfloor(n-\left\lfloor \frac{n}{2} \right\rfloor)} \sum_{i=0}^{\left\lfloor \frac{n}{2} \right\rfloor}
q^{-i(n-2\left\lfloor \frac{n}{2} \right\rfloor+i)} < 2 K_q^{-1} \theta(q) q^{\left\lfloor \frac{n}{2} \right\rfloor(n-\left\lfloor \frac{n}{2} \right\rfloor)}$.
\end{proof}
We observe that by (\ref{eq:Gaussian}) and Lemma~\ref{lemma:bounds_E}, $|E_r(q,n)|$ is the same as $|E(q,n)|$ up to a scalar when $r=\left\lfloor \frac{n}{2} \right\rfloor$ or $r=n-\left\lfloor \frac{n}{2} \right\rfloor$. That is, the volume of $E_{\left\lfloor \frac{n}{2} \right\rfloor}(q,n)$, which is equal to that of $E_{n-\left\lfloor \frac{n}{2} \right\rfloor}(q,n)$ when $\left\lfloor \frac{n}{2} \right\rfloor \neq n-\left\lfloor \frac{n}{2} \right\rfloor$, dominates the volumes of other Grassmannians. This geometric property has significant implication to the packing properties of subspace codes.

We now determine the number of subspaces at a given subspace
distance from a fixed subspace. Let us denote the number of subspaces with dimension $s$ at subspace
distance $d$ from a subspace with dimension $r$ as  $\Ns(r,s,d)$.

\begin{lemma}\label{lemma:Ns}
$\Ns(r,s,d)$ is given by
$q^{u(d-u)} {r \brack u} {n-r \brack d-u}$ when $u=
\frac{r+d-s}{2}$ is an integer, and $0$ otherwise.
\end{lemma}

\begin{proof}
For $U \in E_r(q,n)$ and $V \in E_s(q,n)$, $\ds(U,V) = d$ if and
only if $\dim(U \cap V) = r-u$. Thus there are ${r \brack u}$
choices for $U \cap V$. The subspace $V$ is then completed in
$q^{u(d-u)} {n-r \brack d-u}$ ways.
\end{proof}

We remark that this result in Lemma~\ref{lemma:Ns} is implicitly contained in \cite[Theorem
5]{etzion_isit08} without an explicit proof. It is formally stated here because it is important to the results in this paper. We also denote the
volume of a ball with subspace radius $t$ around a subspace with
dimension $r$ as $\Vs(r,t) \df \sum_{d=0}^t \sum_{s=0}^n
\Ns(r,s,d)$.

We now derive bounds on the volume of a ball with subspace radius.
Since $\Vs(r,t) = \Vs(n-r,t)$ for all $r$ and $t$, we only consider
$r \leq \left\lfloor \frac{n}{2} \right\rfloor$. Also, we assume $t \leq \left\lfloor \frac{n}{2} \right\rfloor$, for only this case will
be needed in this paper.

\begin{proposition}\label{prop:bound_Vs}
For all $q$, $n$, $r \leq \left\lfloor \frac{n}{2} \right\rfloor$, and $t \leq \left\lfloor \frac{n}{2} \right\rfloor$, $q^{-\frac{3}{4}}
q^{g(r,t)} \leq \Vs(r,t) \leq 2\theta(q^3)K_q^{-2} (1+q^{-\frac{4}{3}})\theta(q^{\frac{3}{4}})
q^{g(r,t)}$, where
\begin{equation} \nonumber 
    g(r,t) = \left\{ \begin{array}{ll}
    t(n-r-t) & \mbox{for}\,\, t \leq \frac{n-2r}{3},\\
    \frac{1}{12}(n-2r)^2 + \frac{1}{4}t(2n-t) & \mbox{for}\,\,
    \frac{n-2r}{3} < t \leq \frac{n+4r}{3},\\
    (t-r)(n-t+r) & \mbox{for}\,\, \frac{n+4r}{3} < t \leq
    \frac{n}{2}.
    \end{array} \right.
\end{equation}
\end{proposition}

The proof of Proposition~\ref{prop:bound_Vs} is given in
Appendix~\ref{app:prop:bound_Vs}. We remark that the lower and upper bounds on $\Vs(r,t)$  in
Proposition~\ref{prop:bound_Vs} are tight up to a scalar, and that
$g(r, t)$ depends on both $r$ and $t$. We also observe that $g(r,
t)$ decreases with $r$ for $r \leq \left\lfloor \frac{n}{2} \right\rfloor$. That is, the volume of a
ball around a subspace of dimension $r$ ($r \leq \left\lfloor \frac{n}{2} \right\rfloor$) \textbf{decreases}
with $r$. This observation is significant to the covering properties
of subspace codes with the subspace metric. Figure~\ref{fig:Vs}, where we show $\log_2\left[\Vs(q,n,r,t)\right]$ for $q=2$, $n=10$, $0 \leq r \leq 5$, and $0 \leq t \leq 5$, illustrates this observation.

\begin{figure}
\begin{center}
	\includegraphics[scale=0.65]{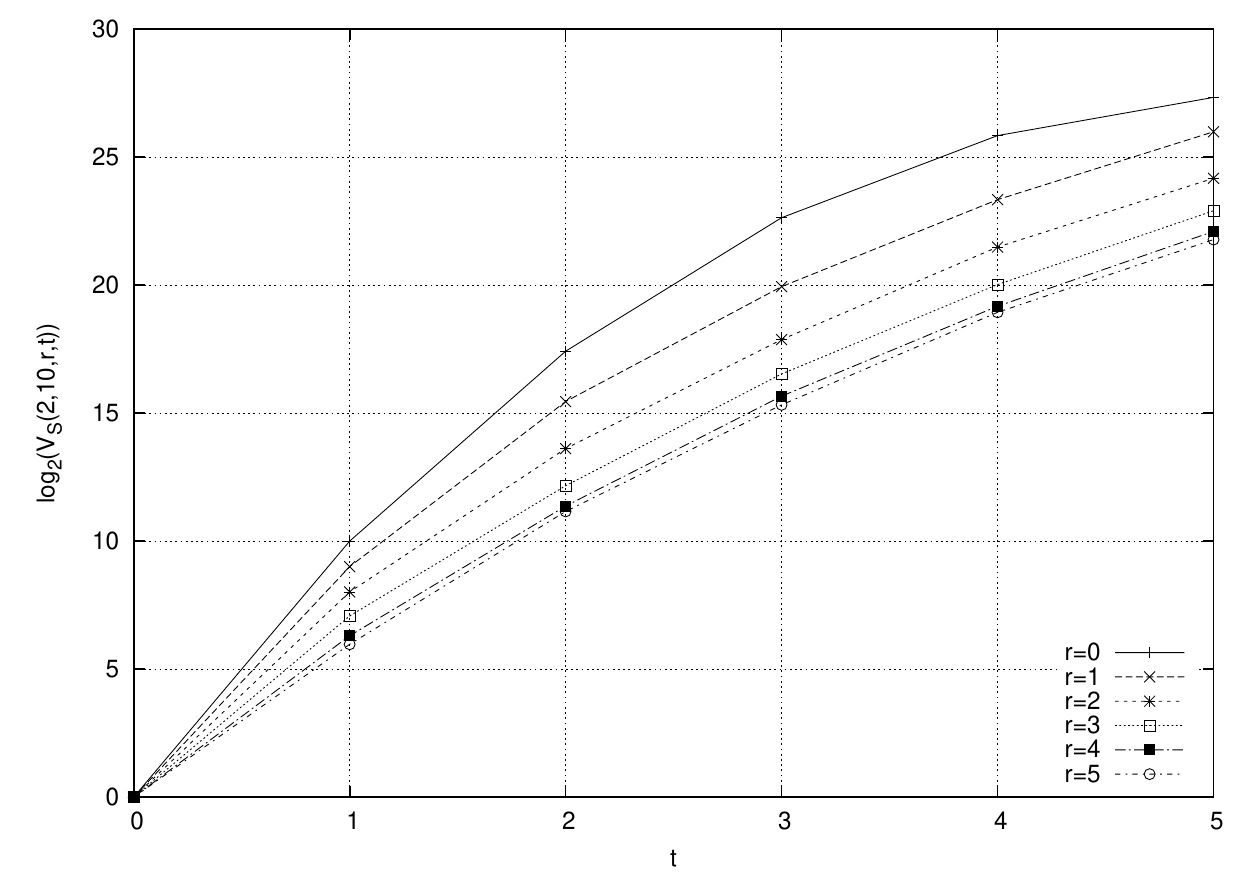}
\end{center}
\caption{Volume of a ball of subspace radius in $E(2,10)$ as a function of the dimension of its center and of its radius} \label{fig:Vs}
\end{figure}

\subsection{Packing properties of subspace codes with the subspace metric}
\label{sec:packing_subspace}

We are interested in packing subspace codes used with the subspace
metric. The maximum cardinality of a code in $E(q,n)$ with minimum
subspace distance $d$ is denoted as $\As(q,n,d)$. Since $\As(q,n,1)
= |E(q,n)|$, we assume $d \geq 2$ henceforth.

We can relate $\As(q,n,d)$ to $\Ac(q,n,r,d)$. First, we remark that $\max_{0 \leq r \leq n} \Ac(q,n,r,d) = \Ac(q,n,\left\lfloor \frac{n}{2} \right\rfloor,d)$ for all $q$, $n$, and $d \leq \left\lfloor \frac{n}{2} \right\rfloor$. The claim is obvious for $d=1$, and easily shown for $d>1$ by using (\ref{eq:Gaussian}). We also remark that $\Ac(q,n,\left\lceil \frac{n}{2} \right\rceil,d) = \Ac(q,n,\left\lfloor \frac{n}{2} \right\rfloor,d)$. For all $J \subseteq \{0,1,\ldots,n\}$, we denote the maximum cardinality of a code with minimum subspace distance $d$ and codewords having dimensions in $J$ as $\As(q,n,d,J)$. For $2 \leq d \leq 2\left\lfloor \frac{n}{2} \right\rfloor$, $R_d \df \left\{\left\lceil \frac{d}{2} \right\rceil, \left\lceil \frac{d}{2} \right\rceil+1, \ldots, n-\left\lceil \frac{d}{2} \right\rceil \right\}$. Proposition \ref{prop:As_Rd} below compares $\As(q,n,d)$ to $\Ac(q,n,\left\lfloor \frac{n}{2} \right\rfloor,d)$ and shows that $\As(q,n,d,R_d)$ is a good approximate of $\As(q,n,d)$.

\begin{proposition}\label{prop:As_Rd}
For $n = d = 2\left\lfloor \frac{n}{2} \right\rfloor+1$, $\As(q,2\left\lfloor \frac{n}{2} \right\rfloor+1,2\left\lfloor \frac{n}{2} \right\rfloor+1) = 2$ and for $2 \leq d \leq 2\left\lfloor \frac{n}{2} \right\rfloor$, $\As(q,n,d) \leq \As(q,n,d,R_d) + 2$. Also, we have
$\Ac(q,n,\left\lfloor \frac{n}{2} \right\rfloor,\left\lceil \frac{d}{2} \right\rceil) \leq \As(q,n,d) \leq
2 + \sum_{r \in R_d} \Ac(q,n,r,\left\lceil \frac{d}{2}
\right\rceil)$.
\end{proposition}

\begin{proof}
Let $\mathcal{C}$ be a code in $E(q,n)$ with minimum subspace
distance $d$. For $C,D \in \mathcal{C}$, we have $\dim(C) + \dim(D)
\geq \ds(C,D) \geq d$; therefore there is at most one codeword with
dimension less than $\frac{d}{2}$. Similarly, $ \dim(C) + \dim(D)
\leq 2n - \ds(C,D) \leq 2n-d$, therefore there is at most one
codeword with dimension greater than $\frac{2n-d}{2}$. Thus
$\As(q,n,d) \leq \As(q,n,d,R_d) + 2$ for $d \leq 2 \left\lfloor \frac{n}{2} \right\rfloor$ and
$\As(q,2\left\lfloor \frac{n}{2} \right\rfloor+1,2\left\lfloor \frac{n}{2} \right\rfloor+1) \leq 2$. Since the code $\{ \{{\bf 0}\},
\mathrm{GF}(q)^{2 \left\lfloor \frac{n}{2} \right\rfloor+1} \}$ has minimum subspace distance $2
\left\lfloor \frac{n}{2} \right\rfloor+1$, we obtain $\As(q,2\left\lfloor \frac{n}{2} \right\rfloor+1,2\left\lfloor \frac{n}{2} \right\rfloor+1) = 2$.

A CDC in $E_r(q,n)$ with minimum injection distance $\left\lceil
\frac{d}{2} \right\rceil$ has minimum subspace distance $\geq d$,
and hence $\Ac(q,n,r,\left\lceil \frac{d}{2} \right\rceil) \leq
\As(q,n,d)$ for all $r$. Also, the codewords with dimension $r$ in a
code with minimum subspace distance $d$ form a CDC in $E_r(q,n)$
with minimum injection distance at least $\left\lceil \frac{d}{2}
\right\rceil$, and hence $\As(q,n,d) \leq \As(q,n,d,R_d) + 2 \leq 2
+ \sum_{r \in R_d} \Ac(q,n,r,\left\lceil \frac{d}{2} \right\rceil)$.
\end{proof}

We compare our lower bound on $\As(q,n,d)$ in Proposition
\ref{prop:As_Rd} to the Gilbert bound in \cite[Theorem
5]{etzion_isit08}. The latter shows that $\As(q,n,d) \geq
\frac{|E(q,n)|}{\mbox{avg} \left\{\Vs(r,d-1)\right\}}$, where the average volume $\mbox{avg} \left\{\Vs(r,d-1)\right\}$ is taken
over all subspaces in $E(q,n)$. Using the bounds on $\Vs(r,d-1)$ in
Proposition~\ref{prop:bound_Vs}, it can be shown that this lower
bound is at most $2 K_q^{-1} \theta(q) q^{\frac{3}{4}} q^{\left\lfloor \frac{n}{2} \right\rfloor(n-\left\lfloor \frac{n}{2} \right\rfloor) -
\frac{1}{4}(d-1)(2n-d+1)}$. On the other hand,
Proposition~\ref{prop:As_Rd} and (\ref{eq:bounds_Ac}) yield
$\As(q,n,d) \geq q^{\left\lfloor \frac{n}{2} \right\rfloor(n-\left\lfloor \frac{n}{2} \right\rfloor) - \frac{1}{4}(d-1)(n+1)}$. The ratio
between our lower bound and the Gilbert bound is hence at least
$\frac{1}{2} K_q \theta(q)^{-1} q^{\frac{1}{4}(d-1)(n-d)}\geq 1$ for all $n$ and $d$. Therefore, our
lower bound in Proposition \ref{prop:As_Rd} is tighter than the
Gilbert bound in \cite[Theorem
5]{etzion_isit08}.

The lower bound in Proposition~\ref{prop:As_Rd} is further tightened
below by considering the union of CDCs in different Grassmannians.

\begin{proposition} \label{prop:lower_As}
For all $q$, $n$, and $2 \leq d \leq n$, we have $\As(q,n,d) \geq \sum_{i=-z}^z \Ac(q,n,\left\lfloor \frac{n}{2} \right\rfloor - id,\left\lceil \frac{d}{2} \right\rceil)$, where $z = \left\lfloor \frac{\left\lfloor \frac{n}{2} \right\rfloor}{d} \right\rfloor$.
\end{proposition}

\begin{proof}
For $i=-z,-z+1,\ldots,z$, let $\mathcal{C}_i$ be a CDC in $E_{\left\lfloor \frac{n}{2} \right\rfloor - id}(q,n)$ with minimum subspace distance $2\left\lceil \frac{d}{2} \right\rceil$ and cardinality $\Ac(q,n,\left\lfloor \frac{n}{2} \right\rfloor - id,\left\lceil \frac{d}{2} \right\rceil)$ and let $\mathcal{C} = \bigcup_{i=-z}^z \mathcal{C}_i$. We have $|\mathcal{C}| = \sum_{i=-z}^z |\mathcal{C}_i|$, and we now prove that $\mathcal{C}$ has minimum subspace distance at least $d$ by considering two distinct codewords $C_i^j \in \mathcal{C}_i$ and $C_a^b \in \mathcal{C}_a$. First, if $i \neq a$, then $\ds(C_i^j, C_a^b) \geq |i-a| d \geq d$; second, if $i=a$ and $j \neq b$, then $\ds(C_a^j, C_a^b) \geq 2\left\lceil \frac{d}{2} \right\rceil$ by the minimum distance of $\mathcal{C}_a$.
\end{proof}

In order to characterize the rate loss by using CDCs instead of subspace codes, we now compare the cardinalities of optimal subspace codes and optimal CDCs with the same minimum subspace distance $d$. Note that the bounds on the cardinalities of optimal CDCs in (\ref{eq:bounds_Ac}) assume the injection metric for CDC.  When $d$ is even, a CDC with a minimum subspace distance $d$ has a minimum injection distance $\frac{d}{2}$. When $d$ is odd, a CDC with a minimum subspace distance $d+1$ has a minimum injection distance $\frac{d+1}{2}=\left\lceil
\frac{d}{2} \right\rceil$. Thus, a CDC has a minimum subspace distance at least $d$ if
and only if it has minimum injection distance at least $\left\lceil
\frac{d}{2} \right\rceil$. Hence, we compare $\As(q,n,d)$ and
$\Ac(q,n,r,\left\lceil \frac{d}{2} \right\rceil)$
in Proposition~\ref{prop:As_v_Ac} below.

\begin{proposition} \label{prop:As_v_Ac}
\textit{(Comparison between optimal subspace codes and CDCs in the subspace metric).} For $2 \leq d \leq 2\left\lfloor \frac{n}{2} \right\rfloor$ and $\left\lceil\frac{d}{2}\right\rceil
\leq r \leq \left\lfloor \frac{n}{2} \right\rfloor$,
\begin{align}
	\nonumber
    & K_q q^{(\left\lfloor \frac{n}{2} \right\rfloor-r)(\left\lfloor \frac{n}{2} \right\rfloor - r + \left\lceil\frac{d}{2}\right\rceil - 1)}
    \Ac \left(q,n,r,\left\lceil\frac{d}{2}\right\rceil \right)\\
    \nonumber
    & < \As(q,n,d)\\
    \nonumber
    & < 2 K_q^{-1} \theta(q) q^{(\left\lfloor \frac{n}{2} \right\rfloor-r)(\left\lfloor \frac{n}{2} \right\rfloor - r + \left\lceil\frac{d}{2}\right\rceil - 1)}
    \Ac \left(q,n,r,\left\lceil\frac{d}{2}\right\rceil \right).\\
    \label{eq:As_v_Ac}
\end{align}
\end{proposition}

\begin{proof}
By (\ref{eq:Gaussian}), Proposition~\ref{prop:As_Rd}, and
(\ref{eq:bounds_Ac}), we have $\As(q,n,d) \geq
\Ac(q,n,\left\lfloor \frac{n}{2} \right\rfloor,\left\lceil\frac{d}{2}\right\rceil) \geq
q^{(n-\left\lfloor \frac{n}{2} \right\rfloor)(\left\lfloor \frac{n}{2} \right\rfloor-\left\lceil \frac{d}{2} \right\rceil+1)} > K_q
q^{(\left\lfloor \frac{n}{2} \right\rfloor-r)(\left\lfloor \frac{n}{2} \right\rfloor - r + \left\lceil\frac{d}{2}\right\rceil - 1)}
    \Ac \left(q,n,r,\left\lceil\frac{d}{2}\right\rceil \right)$.
Also, Proposition~\ref{prop:As_Rd} and (\ref{eq:Gaussian}) also lead to
\begin{align}
    \nonumber
    \As(q,n,d) <& 2 + 2 K_q^{-1} \sum_{r = \left\lceil \frac{d}{2}
    \right\rceil}^{\left\lfloor \frac{n}{2} \right\rfloor} q^{(n-r)(r-\left\lceil \frac{d}{2}
    \right\rceil+1)}\\
    \nonumber
    =& 2 + 2 K_q^{-1} q^{(n-\left\lfloor \frac{n}{2} \right\rfloor)(\left\lfloor \frac{n}{2} \right\rfloor-\left\lceil \frac{d}{2}
    \right\rceil+1)}\\
	\nonumber    
    &\sum_{i=0}^{\left\lfloor \frac{n}{2} \right\rfloor-\left\lceil \frac{d}{2} \right\rceil} q^{-i(n-2\left\lfloor \frac{n}{2} \right\rfloor+ \left\lceil \frac{d}{2}   	 \right\rceil - 1 + i)},
\end{align}
where $i = \left\lfloor \frac{n}{2} \right\rfloor - r$. Since $n-2\left\lfloor \frac{n}{2} \right\rfloor \geq 0$ and $\left\lceil \frac{d}{2} \right\rceil - 1 \geq 0$, we obtain
\begin{align}
    \nonumber
    \As(q,n,d) <& 2 + 2 K_q^{-1} \theta(q) q^{(n-\left\lfloor \frac{n}{2} \right\rfloor)(\left\lfloor \frac{n}{2} \right\rfloor-\left\lceil
    \frac{d}{2} \right\rceil+1)}\\
    \nonumber
    \leq& 2 K_q^{-1} \theta(q) q^{(\left\lfloor \frac{n}{2} \right\rfloor-r)(\left\lfloor \frac{n}{2} \right\rfloor - r + \left\lceil\frac{d}{2}\right\rceil - 1)}\\
    \label{eq:As1}
	& \Ac \left(q,n,r, \left\lceil \frac{d}{2} \right\rceil \right),
\end{align}
where (\ref{eq:As1}) follows from (\ref{eq:bounds_Ac}).
\end{proof}

We now compare the relation between $\As(q,n,d)$ and $\Ac(q,n,r,d)$ in Proposition \ref{prop:As_v_Ac} to the one determined in \cite[Theorem 5]{etzion_it09}. The latter only provides the following lower bound on $\As(q,n,d)$: $\As(q,n,d) \geq \frac{q^{n+1-r} + q^r - 2}{q^{n+1} - 1} \Ac(q,n+1,r,\left\lceil \frac{d}{2} \right\rceil+1)$. The Singleton bound on CDCs \cite{koetter_it08} indicates that $\Ac(q,n+1,r,d+1) \leq \Ac(q,n,r,d-1)$, which in turn satisfies $\Ac(q,n,r,d-1) < K_q^{-1} q^{-(n-2r-d)} \Ac(q,n,r,d)$ by (\ref{eq:Gaussian}). Hence the lower bound on $\As(q,n,d)$ in \cite[Theorem 5]{etzion_it09} is at most $2 K_q^{-1} q^{-(n-r+\left\lceil \frac{d}{2} \right\rceil-1)} \Ac(q,n,r,\left\lceil \frac{d}{2} \right\rceil)$. The ratio between our lower bound in Proposition \ref{prop:As_v_Ac} and the lower bound in \cite[Theorem 5]{etzion_it09} is at least $\frac{K_q}{2} q^{n-\left\lfloor \frac{n}{2} \right\rfloor + (\left\lfloor \frac{n}{2} \right\rfloor-r+1)(\left\lfloor \frac{n}{2} \right\rfloor-r+\left\lceil \frac{d}{2} \right\rceil-1)}$, and thus our lower bound in Proposition \ref{prop:As_v_Ac} is tighter than the bound in \cite[Theorem 5]{etzion_it09} for all cases.

The bounds in Proposition~\ref{prop:As_v_Ac} help us determine the
asymptotic behavior of $\As(q,n,d)$. We first define the rate of a
subspace code $\mathcal{C} \subseteq E(q,n)$ as $\frac{\log_q
|\mathcal{C}|}{\log_q |E(q,n)|}$. We note that this definition is
combinatorial, and differs from the rate introduced in
\cite{koetter_it08} for CDCs. The rate defined in
\cite{koetter_it08} also accounts for the channel usage,
but it seems appropriate for CDCs only. On the other
hand, our rate depends on only the cardinality of the code, and
hence is more appropriate to compare general subspace codes, since
all the subspaces are treated equally regardless of their
dimension. Finally, the rate defined in \cite{koetter_it08} can be
derived from our rate defined here. Using the normalized parameters
$r' \df \frac{r}{n}$ and $\ds' \df \frac{\ds}{n}$ where $\ds$
is the minimum subspace distance of a code, the asymptotic rate of a
subspace code $\as(\ds') = \lim \sup_{n \rightarrow \infty}
\frac{\log_q \As(q,n,\lceil n\ds' \rceil)}{\log_q |E(q,n)|}$ and of a CDC of given
dimension $\as(r',\ds') = \lim \sup_{n \rightarrow \infty}
\frac{\log_q \Ac(q,n,nr',\left\lceil n\frac{\ds'}{2}
\right\rceil)}{\log_q |E(q,n)|}$ can be easily determined.

\begin{proposition}\label{prop:as}
\textit{(Asymptotic rate of packing subspace codes in the subspace metric).}
For $0 \leq \ds' \leq 1$, $\as(\ds') = 1-\ds'$. For $0 \leq r' \leq
\frac{\ds'}{2}$ or $1-\frac{\ds'}{2} \leq r' \leq 1$, $\as(r',\ds')
= 0$; for $\frac{\ds'}{2} \leq r \leq \frac{1}{2}$, $\as(r',\ds') =
2(1-r')(2r'-\ds')$; for $\frac{1}{2} \leq r' \leq 1-\frac{\ds'}{2}$,
$\as(r',\ds') = 2r'(2-2r'-\ds')$.
\end{proposition}

\begin{proof}
First, (\ref{eq:bounds_Ac}) and Lemma \ref{lemma:bounds_E} yield
$\as(r',\ds') = 2(1-r')(2r'-\ds')$ for $0 \leq r' \leq
\frac{\ds'}{2}$. Since $\Ac(q,n,r,\left\lceil \frac{d}{2}
\right\rceil) = \Ac(q,n,n-r,\left\lceil \frac{d}{2} \right\rceil)$,
we also obtain $\as(r',\ds') = 2(1-r')(2r'-\ds')$ for
$\frac{\ds'}{2} \leq r \leq \frac{1}{2}$. Second, (\ref{eq:As_v_Ac})
for $r = \left\lfloor \frac{n}{2} \right\rfloor$ and (\ref{eq:bounds_Ac}) yield $\as(\ds') =
\as(\frac{1}{2}, \ds') = 1-\ds'$.
\end{proof}

Propositions~\ref{prop:As_v_Ac} and \ref{prop:as} provide several
important insights. First, Proposition~\ref{prop:As_v_Ac} indicates
that optimal CDCs with dimension being half of the block length up
to rounding ($r = \left\lfloor \frac{n}{2} \right\rfloor$ and $r =n- \left\lfloor \frac{n}{2} \right\rfloor$)  are optimal subspace codes
up to a scalar. In this case, the optimal CDCs have a limited rate
loss as opposed to optimal subspace codes with the same error
correction capability. When $r < \left\lfloor \frac{n}{2} \right\rfloor$, the rate loss suffered by
optimal CDCs increases with $\left\lfloor \frac{n}{2} \right\rfloor - r$. Proposition~\ref{prop:as}
indicates that using CDCs with dimension $\left\lceil \frac{\ds}{2}
\right\rceil \leq r < \left\lfloor \frac{n}{2} \right\rfloor$ leads to a decrease in rate on the order
of $(1-2r')(\ds'+1-2r')$, where $r' = \frac{r}{n}$. Since the rate
loss increases with $1-2r'$, using a CDC with a dimension further
from $\left\lfloor \frac{n}{2} \right\rfloor$ leads to a larger rate loss.

The conclusion above can be explained from a combinatorial
perspective as well. When $r = \left\lfloor \frac{n}{2} \right\rfloor$ or $r =n- \left\lfloor \frac{n}{2} \right\rfloor$, by
Lemma~\ref{lemma:bounds_E},  $|E(q,n)|$ is the same as
$|E_r(q,n)|={n \brack r}$ up to scalar. Thus it is not surprising
that the optimal packings in $E(q,n)$ are the same as those in
$E_r(q,n)$ up to scalar.

We also comment that the asymptotic rates in
Proposition~\ref{prop:as} for subspace codes come from Singleton
bounds. The asymptotic rate $\as(r',\ds')$ is achieved by KK codes.
The asymptotic rate $\as(\ds')$ is similar to that for rank metric
codes \cite{gadouleau_it08_covering}. This can be explained by the
fact that the asymptotic rate $\as(\ds')$ is also achieved by KK
codes when $r = \left\lfloor \frac{n}{2} \right\rfloor$, whose cardinalities are equal to those of
optimal rank metric codes.

In Table \ref{table:comparison} we compare the bounds on $\As(q,n,d)$ derived in this paper with each other and with existing bounds in the literature, for $q=2$, $n=10$, and $d$ ranging from $2$ to $10$. We consider the lower bound in Proposition \ref{prop:As_Rd}, its refinement in Proposition \ref{prop:lower_As}, and the lower bounds in \cite{etzion_isit08} and \cite[Theorem 5]{etzion_it09} described above, and the upper bound comes from Proposition \ref{prop:As_Rd}. Note that Proposition~\ref{prop:As_v_Ac} is not included in the comparison since its purpose is to compare the cardinalities of optimal subspace codes and optimal CDCs with the same minimum subspace distance. Since bounds in Propositions~\ref{prop:As_Rd} and \ref{prop:lower_As} and \cite[Theorem~5]{etzion_it09} depend on cardinalities of either related CDCs or optimal CDCs, we use the cardinalities of CDCs with dimension $r=n/2=5$ proposed in \cite{etzion_it09} and \cite{gadouleau_it09_cdc} as lower bounds on $\Ac(q,n,r,d)$ and the upper bound in \cite{xia_dcc09} on $\Ac(q,n,r,d)$ to derive the numbers in Table \ref{table:comparison}. For example, the lower bound of Proposition \ref{prop:As_Rd} is simply given by the construction in \cite{etzion_it09} when $d=3, 4, 5, \mbox{ and } 6$,  and given by the construction in \cite{gadouleau_it09_cdc} for other values of $d$.
Table~\ref{table:comparison} illustrates our lower bounds in Propositions~\ref{prop:As_Rd} and \ref{prop:lower_As} are tighter than those in \cite{etzion_isit08} and \cite[Theorem 5]{etzion_it09}. The cardinalities of CDCs with dimension $r=n/2$ in \cite{etzion_it09} and \cite{gadouleau_it09_cdc}, displayed as the lower bound in Proposition \ref{prop:As_Rd}, are quite close to the lower bound in Proposition \ref{prop:lower_As}, supporting our conclusion that the rate loss suffered by properly designed CDCs is smaller when the dimension is close to $n/2$. Also, the lower and upper bounds in Proposition~\ref{prop:As_Rd} depend on $\left\lceil \frac{d}{2} \right\rceil$, and hence the bounds for $d=2l$ and $d=2l-1$ are the same. Finally, the tightness of the bounds improves as the minimum distance of the code increases, leading to very tight bounds for $d=n$.

\begin{table*}
\begin{center}
\begin{tabular}{|c|c|c|c|c||c|}
	\hline
	 & \multicolumn{4}{|c||}{Lower bounds} & Upper bound\\
	\hline
	$d$ & Proposition \ref{prop:As_Rd} & Proposition \ref{prop:lower_As} & \cite{etzion_isit08} & \cite[Theorem 5]{etzion_it09} & Proposition \ref{prop:As_Rd}\\
	\hline
	2 & 52,494,849 & 59,058,177 & 3,181,506 & 3,073,032 & 229,755,605\\
	3 & 1,167,327 & 1,167,967 & 64,047 & 88,163 & 2,616,760\\
	4 & 1,167,327 & 1,167,329 & 64,047 & 4,650 & 2,616,760\\
	5 & 32,841 & 32,843 & 1,986 & 397 & 50,708\\
	6 & 32,841 & 32,841 & 1,986 & 54 & 50,708\\
	7 & 1,025 & 1,025 & 63 & 12 & 1,260\\
	8 & 1,025 & 1,025 & 63 & 4 & 1,260\\
	9 & 33 & 33 & 2 & 2 & 35\\
	10 & 33 & 33 & 2 & 2 & 35\\
	\hline
\end{tabular} \caption{Comparison of bounds on $\As(2,10,d)$ for $d$ from $2$ to $10$} \label{table:comparison}
\end{center}
\end{table*}

\subsection{Covering properties of subspace codes with the subspace metric}
\label{sec:covering_subspace}

We now consider the covering properties of subspace codes with the
subspace metric. The subspace covering radius in $E(q,n)$ of a code $\mathcal{C}$
 is defined as $\max_{U \in E(q,n)}
\ds(U,\mathcal{C})$. We denote the minimum cardinality of a subspace
code in $E(q,n)$ with subspace covering radius $\rho$ as
$\Ks(q,n,\rho)$. Since $\Ks(q,n,0) = |E(q,n)|$ and $\Ks(q,n,n) = 1$,
we assume $0 < \rho < n$ henceforth. We determine below the minimum
cardinality of a code with subspace covering radius $\rho \geq \left\lfloor \frac{n}{2} \right\rfloor$.

\begin{proposition}\label{prop:Ks_2rho>n}
For $\left\lfloor \frac{n}{2} \right\rfloor \leq \rho < n$, $\Ks(q,n,\rho) = 2$.
\end{proposition}

\begin{proof}
For all $V \in E(q,n)$ there exists ${\bar V}$ such that $V \cap
\bar{V} = \{{\bf 0}\}$ and $V + {\bar V} = \mathrm{GF}(q)^n$, and
hence $\ds(V,{\bar V}) = n$. Therefore, one subspace cannot cover
the whole $E(q,n)$ with radius $\rho < n$, hence $\Ks(q,n,\rho) >
1$. Let $\mathcal{C} = \{ \{{\bf 0}\}, \mathrm{GF}(q)^n \}$, then
for all $D \in E(q,n)$, $\ds(D, \mathcal{C}) = \min\{ \dim(D),
n-\dim(D) \} \leq \left\lfloor \frac{n}{2} \right\rfloor$. Thus $\mathcal{C}$ has covering radius $\left\lfloor \frac{n}{2} \right\rfloor$
and $\Ks(q,n,\rho) \leq 2$ for all $\rho \geq \left\lfloor \frac{n}{2} \right\rfloor$.
\end{proof}

We thus consider $0 < \rho < \left\lfloor \frac{n}{2} \right\rfloor$ henceforth.
Proposition~\ref{prop:sphere_bound_subspace} below can be viewed as the sphere
covering bound for subspace codes with the subspace metric, as it considers how a subspace code covers each Grassmannian $E_r(q,n)$ for any $0 \leq r \leq n$.

\begin{proposition} \label{prop:sphere_bound_subspace}
\textit{(Sphere covering bound for the subspace metric).}
For all $q$, $n$, and $0 < \rho < \left\lfloor \frac{n}{2} \right\rfloor$, $\Ks(q,n,\rho) \geq \min
\sum_{i=0}^n A_i$, where the minimum is taken over all integer
sequences $\{A_i\}$ satisfying $0 \leq A_i \leq {n \brack i}$ for
all $0 \leq i \leq n$ and $\sum_{i=0}^n A_i \sum_{d=0}^\rho
\Ns(i,r,d) \geq {n \brack r}$ for $0 \leq r \leq n$.
\end{proposition}

\begin{proof}
Let $\mathcal{C}$ be a subspace code with covering radius $\rho$ and
let $A_i$ denote the number of subspaces with dimension $i$ in
$\mathcal{C}$. Then $0 \leq A_i \leq {n \brack i}$ for all $0 \leq i
\leq n$. All subspaces with dimension $r$ are covered; however, a
codeword with dimension $i$ covers exactly $\sum_{d=0}^\rho
\Ns(i,r,d)$ subspaces with dimension $r$, hence $\sum_{i=0}^n A_i
\sum_{d=0}^\rho \Ns(i,r,d) \geq {n \brack r}$ for $0 \leq r \leq n$.
\end{proof}

We remark that the lower bound in Proposition
\ref{prop:sphere_bound_subspace} is based on the optimal solution to an
integer linear program and hence determining this lower bound is
computationally infeasible for large parameter values.

We now derive upper bounds on $\Ks(q,n,\rho)$. Since $|E_{\left\lfloor \frac{n}{2} \right\rfloor}(q,n)|$ is equal to $E(q,n)$ up to a scalar, the main issue with designing covering subspace codes is to cover $E_{\left\lfloor \frac{n}{2} \right\rfloor}(q,n)$. In Proposition \ref{prop:greedy_covering_subspace}, we use subspaces in $E_r(q,n)$ in order to cover the Grassmannian $E_{r+\rho}(q,n)$ for $r \leq \left\lfloor \frac{n}{2} \right\rfloor$, i.e., $E_{\left\lfloor \frac{n}{2} \right\rfloor}(q,n)$ is covered using subspaces in $E_{\left\lfloor \frac{n}{2} \right\rfloor - \rho}(q,n)$. This choice is in fact asymptotically optimal, as we shall show in Proposition \ref{prop:ks}.

The upper bound in Proposition \ref{prop:greedy_covering_subspace} below is based on the universal greedy algorithm in \cite[Theorem 12.2.1]{cohen_book97} to construct covering codes, which we briefly review below for subspaces. The algorithm begins by selecting as the first codeword one of the subspaces which cover the most subspaces, and then keeps adding subspaces to the code. Each new codeword is selected as to cover the most subspaces not yet covered by the code (if several subspaces cover the same number of subspaces, then the new codeword is chosen randomly). The algorithm eventually stops once all subspaces are covered. Although the cardinality of the code obtained by this algorithm is not constant, an upper bounded on its value is given in \cite[Theorem 12.2.1]{cohen_book97}. The bound in Proposition \ref{prop:greedy_covering_subspace} adapts this algorithm to cover each Grassmannian $E_{r+\rho}(q,n)$ for $r \leq \left\lfloor \frac{n}{2} \right\rfloor$ by subspaces in $E_r(q,n)$. We remark that the bound in Proposition \ref{prop:greedy_covering_subspace} is only semi-constructive, as it determines an algorithm to construct covering subspace codes but does not design the actual codes. We remark that the bound in Proposition~\ref{prop:greedy_covering_subspace} can be further tightened by using the bounds on the greedy algorithm derived in \cite{clark_ejc97, gadouleau_cl08}.

\begin{proposition}\label{prop:greedy_covering_subspace}
For all $q$, $n$, $0 < \rho < \left\lfloor \frac{n}{2} \right\rfloor$, $\Ks(q,n,\rho) \leq 2 + 2
\sum_{r=\rho+1}^{\left\lfloor \frac{n}{2} \right\rfloor} \left\lfloor k_r \right \rfloor$, where
\begin{equation} \nonumber
    k_r = \frac{{n \brack r}}{{n-r+\rho \brack \rho}} +
    \frac{{n \brack r-\rho}}{{r \brack \rho}} \ln {n-r+\rho \brack
    \rho}.
\end{equation}
\end{proposition}

\begin{proof}
We show that there exists a code with cardinality $2 + 2
\sum_{r=\rho+1}^{\left\lfloor \frac{n}{2} \right\rfloor} \left\lfloor k_r \right \rfloor$ and covering
radius $\rho$. We choose $\{{\bf 0}\}$ to be in the code, hence all
subspaces with dimension $0 \leq r \leq \rho$ are covered. For
$\rho+1 \leq r \leq \left\lfloor \frac{n}{2} \right\rfloor$, let ${\bf A}$ be the ${n \brack r} \times
{n \brack r-\rho}$ binary matrix whose rows represent the subspaces
$U_i \in E_r(q,n)$ and whose columns represent the subspaces $V_j
\in E_{r-\rho}(q,n)$, and where $a_{i,j} = 1$ if and only if
$\ds(U_i,V_j) = \rho$. Then there are exactly $\Ns(r,r-\rho,\rho) =
{r \brack \rho}$ ones on each row and $\Ns(r-\rho,r,\rho) =
{n-r+\rho \brack \rho}$ ones on each column. By \cite[Theorem
12.2.1]{cohen_book97}, there exists an ${n \brack r} \times
\left\lfloor k_r \right \rfloor$ submatrix of ${\bf A}$ with no
all-zero rows. Thus, all subspaces of dimension $r$ can be covered
using $\left\lfloor k_r \right \rfloor$ codewords. Summing for all
$r$, all subspaces with dimension $0 \leq r \leq \left\lfloor \frac{n}{2} \right\rfloor$ can be covered
with $1 + \sum_{r=\rho+1}^{\left\lfloor \frac{n}{2} \right\rfloor} \left\lfloor k_r \right \rfloor$
subspaces. Similarly, it can be shown that all subspaces with
dimension $\left\lfloor \frac{n}{2} \right\rfloor+1 \leq r \leq n$ can be covered with $1 +
\sum_{r=\rho+1}^{\left\lfloor \frac{n}{2} \right\rfloor} \left\lfloor k_r \right \rfloor$ subspaces.
\end{proof}

In Proposition \ref{prop:subspace_covering_code} below, we design an explicit construction of a subspace covering code by combining entire Grassmannians.

\begin{proposition}\label{prop:subspace_covering_code}
For all $q$, $n$, and $0 < \rho < \left\lfloor \frac{n}{2} \right\rfloor$, let $J_1 = \{0\} \cup
\{\left\lfloor \frac{n}{2} \right\rfloor-\rho-\lfloor\frac{\left\lfloor \frac{n}{2} \right\rfloor-\rho}{2\rho+1}\rfloor(2\rho+1), \ldots,
\left\lfloor \frac{n}{2} \right\rfloor-3\rho-1, \left\lfloor \frac{n}{2} \right\rfloor-\rho \}$ and $J_2 = \{i : n-i \in J_1\}$.
Then the code $\bigcup_{r \in J_1 \cup J_2} E_r(q,n)$ has subspace covering radius $\rho$, and
hence $\Ks(q,n,\rho) \leq \sum_{r \in J_1 \cup J_2} {n \brack r}$.
\end{proposition}

\begin{proof}
We prove that $\bigcup_{r \in J_1} E_r(q,n)$ covers all subspaces
with dimension $\leq \left\lfloor \frac{n}{2} \right\rfloor$. First, all subspaces $D_0 \in E(q,n)$
with dimension $0 \leq \dim(D_0) < \left\lfloor \frac{n}{2} \right\rfloor - 2\rho -
\lfloor\frac{\left\lfloor \frac{n}{2} \right\rfloor-\rho}{2\rho+1}\rfloor(2\rho+1) \leq \rho$ are
covered by the subspace with dimension $0$. Second, for all $D_1 \in
E(q,n)$ with dimension $\left\lfloor \frac{n}{2} \right\rfloor-2\rho - i(2\rho+1) \leq \dim(D_1) \leq
\left\lfloor \frac{n}{2} \right\rfloor-\rho - i(2\rho+1)$, there exists $C_1$ with dimension $\left\lfloor \frac{n}{2} \right\rfloor-\rho
- i(2\rho+1)$ such that $D_1 \subseteq C_1$. Thus $\ds(C_1,D_1) =
\dim(C_1) - \dim(D_1) \leq \rho$. Similarly, for all $D_2 \in
E(q,n)$ with dimension $\left\lfloor \frac{n}{2} \right\rfloor-\rho - i(2\rho+1) < \dim(D_2) < \left\lfloor \frac{n}{2} \right\rfloor -
2\rho - (i-1)(2\rho+1)$, there exists $C_2$ with dimension $\left\lfloor \frac{n}{2} \right\rfloor-\rho
- i(2\rho+1)$ such that $C_2 \subset D_2$. Thus $\ds(C_2,D_2) =
\dim(D_2) - \dim(C_2) \leq \rho$. Therefore, $\bigcup_{r \in J_1}
E_r(q,n)$ covers all subspaces with dimension $\leq \left\lfloor \frac{n}{2} \right\rfloor$. Similarly,
all the subspaces with dimension $\geq n-\left\lfloor \frac{n}{2} \right\rfloor$ are covered by
$\bigcup_{r \in J_2} E_r(q,n)$.
\end{proof}

Using the bounds derived above, we now determine the asymptotic
behavior of $\Ks(q,n,\rho)$. We define $\ks(\rho') = \lim \inf_{n
\rightarrow \infty} \frac{\log_q \Ks(q,n,\lfloor n\rho' \rfloor)}{\log_q |E(q,n)|}$,
where $\rho' = \frac{\rho}{n}$. We note that this definition of asymptotic rate is from a combinatorial perspective again.

\begin{proposition}\label{prop:ks}
\textit{(Asymptotic rate of covering subspace codes in the subspace metric).}
For $0 \leq \rho' \leq \frac{1}{2}$, $\ks(\rho') = 1 - 2 \rho'$. For
$\frac{1}{2} \leq \rho' \leq 1$, $\ks(\rho') = 0$.
\end{proposition}

\begin{proof}
By Proposition~\ref{prop:Ks_2rho>n}, $\ks(\rho') = 0$ for $\frac{1}{2} \leq \rho' \leq 1$. Let $\mathcal{C}$ be a KK code in $E_{\left\lfloor \frac{n}{2} \right\rfloor}(q,n)$ with minimum subspace distance $2\rho+1$ and cardinality $q^{(n-\left\lfloor \frac{n}{2} \right\rfloor)(\left\lfloor \frac{n}{2} \right\rfloor - 2\rho)}$. Then any code $\mathcal{D} \subseteq E(q,n)$ with subspace covering radius $\rho$ and cardinality $\Ks(q,n,\rho)$ covers all codewords in $\mathcal{C}$; however, any codeword in $\mathcal{D}$ only covers at most one codeword in $\mathcal{C}$. Hence $\Ks(q,n,\rho) \geq q^{(n-\left\lfloor \frac{n}{2} \right\rfloor)(\left\lfloor \frac{n}{2} \right\rfloor - 2\rho)}$, which asymptotically becomes $\ks(\rho') \geq 1 - 2 \rho'$.


Also, by Proposition~\ref{prop:greedy_covering_subspace}, it can be
easily shown that $\Ks(q,n,\rho) \leq 2+(n+1)[1 - \ln K_q +
\rho(n-\rho-1)\ln q] K_q^{-1} q^{(n-\left\lfloor \frac{n}{2} \right\rfloor)(\left\lfloor \frac{n}{2} \right\rfloor-\rho)}$, which
asymptotically becomes $\ks(\rho') \leq 1 - 2 \rho'$.
\end{proof}

The proof of Proposition \ref{prop:ks} indicates that the minimum
cardinality $\Ks(q,n,\rho)$ of a covering subspace code is on the
order of $q^{(n-\left\lfloor \frac{n}{2} \right\rfloor)(\left\lfloor \frac{n}{2} \right\rfloor-\rho)}$. However, a covering subspace code
is easily obtained by taking the union of optimal covering CDCs (in
their respective Grassmannians) for all dimensions, leading to a
code with cardinality $2 + \sum_{r=\rho+1}^{n-\rho-1}
\Kc(q,n,r,\left\lfloor \frac{\rho}{2} \right\rfloor)$. By
\cite[Proposition 11]{gadouleau_it09_cdc}, $\Kc(q,n,r,\left\lfloor
\frac{\rho}{2} \right\rfloor)$ is on the order of $q^{r(n-r) -
\left\lfloor \frac{\rho}{2} \right\rfloor(n-\left\lfloor
\frac{\rho}{2} \right\rfloor)}$. Hence the code has a cardinality on
the order of $q^{\left\lfloor \frac{n}{2} \right\rfloor(n-\left\lfloor \frac{n}{2} \right\rfloor) - \frac{\rho}{2}(n-\frac{\rho}{2})}$,
which is greater than $q^{(n-\left\lfloor \frac{n}{2} \right\rfloor)(\left\lfloor \frac{n}{2} \right\rfloor-\rho)}$. Thus, a union of
optimal covering  CDCs (in their respective Grassmannians) does not
result in asymptotically optimal covering subspace codes with the
subspace metric.

\section{Packing and covering properties of subspace codes with the
injection metric} \label{sec:injection}

\subsection{Properties of balls with injection radii}
\label{sec:balls_injection}

We first investigate the properties of balls with injection radii in
$E(q,n)$, which will be instrumental in our study of packing and
covering properties of subspace codes with the injection distance.
We denote the number of subspaces with dimension $s$ at injection distance $d$
from a subspace with dimension $r$ as $\Ni(r,s,d)$.

\begin{lemma} \label{lemma:Nm}
$\Ni(r,s,d) =\Ns(r,s,2d-|r-s|)$. Hence, $\Ni(r,s,d) = q^{d(d+s-r)} {r \brack d}
{n-r \brack d+s-r}$ for $r \geq s$ and $\Ni(r,s,d) = q^{d(d+r-s)} {r
\brack d+r-s} {n-r \brack d}$ for $r \leq s$.
\end{lemma}

\begin{proof}
If $U \in E_r(q,n)$ and $V \in E_s(q,n)$, then $\di(U,V) = d$ if and
only if $\ds(U,V) = 2d - |r-s|$. Therefore, $\Ni(r,s,d) =
\Ns(r,s,2d-|r-s|)$, and the formula for $\Ni(r,s,d)$ is easily
obtained from Lemma~\ref{lemma:Ns}.
\end{proof}

Lemma \ref{lemma:Nm} indicates that the injection metric satisfies a strengthened triangular inequality: for any $U \in E_r(q,n)$ and $V \in E_s(q,n)$, we have $\di(U,V) \leq \max\{r,s\}$. We denote the volume of a ball with injection radius $t$ around a
subspace with dimension $r$ as $\Vi(r,t) \df \sum_{d=0}^t
\sum_{s=0}^n \Ni(r,s,d)$. Although the volume $\Vi(r,t)$ of a ball depends on its radius $t$ and on the dimension $r$ of its center, we derive below bounds on $\Vi(r,t)$ which only depend on its radius.

\begin{proposition}\label{prop:bound_Vm}
For all $q$, $n$, $r$, and $t \leq \left\lfloor \frac{n}{2} \right\rfloor$, $q^{t(n-t)} \leq \Vi(r,t) <
\theta(q)(2\theta(q)-1)K_q^{-2} q^{t(n-t)}$.
\end{proposition}

The proof of Proposition~\ref{prop:bound_Vm} is given in Appendix~\ref{app:prop:bound_Vm}. We remark that the bounds in Proposition~\ref{prop:bound_Vm} are tight up to a scalar, which will greatly facilitate our asymptotic study of subspace codes with the injection metric. Unlike the bounds on the volume of a ball with subspace radius in Proposition~\ref{prop:bound_Vs}, the lower and upper bounds in Proposition~\ref{prop:bound_Vm} do not depend on $r$. This illustrates a clear geometric distinction between the subspace and injection metrics.

\subsection{Packing properties of subspace codes with the injection metric} \label{sec:packing_injection}
We are interested in packing subspace codes used with the injection
metric. The maximum cardinality of a code in $E(q,n)$ with minimum
injection distance $d$ is denoted as $\Ai(q,n,d)$. Since $\Ai(q,n,1)
= |E(q,n)|$, we assume $d \geq 2$ henceforth. When $d > \left\lfloor \frac{n}{2} \right\rfloor$, the
maximum cardinality of a code with minimum injection distance $d$ is
determined and a code with maximum cardinality is given. For all $J
\subseteq \{0,1,\ldots,n\}$, we denote the maximum cardinality of a
code with minimum injection distance $d$ and codewords having
dimensions in $J$ as $\Ai(q,n,d,J)$. For $2 \leq d \leq \left\lfloor \frac{n}{2} \right\rfloor$, we
denote $Q_d = \{d, d+1, \ldots, n-d\}$. Proposition \ref{prop:Am_2d>n} below relates $\Ai(q,n,d)$ to $\Ac(q,n,\left\lfloor \frac{n}{2} \right\rfloor,d)$ and shows that determining $\Ai(q,n,d,Q_d)$ is equivalent to determining $\Ai(q,n,d)$.

\begin{proposition}\label{prop:Am_2d>n}
For $d > \left\lfloor \frac{n}{2} \right\rfloor$, $\Ai(q,n,d) = 2$ and for $2 \leq d \leq \left\lfloor \frac{n}{2} \right\rfloor$,
$\Ai(q,n,d) = \Ai(q,n,d,Q_d) + 2$.
\end{proposition}

\begin{proof}
Let $\mathcal{C}$ be a code in $E(q,n)$ with minimum injection
distance $d$ and let $C,D \in \mathcal{C}$. We have $\max\{\dim(C),
\dim(D)\} = \di(C,D) + \dim(C \cap D) \geq d$, therefore there is at
most one codeword with dimension less than $d$. Also, $\min\{
\dim(C), \dim(D) \} = \dim(C+D) - \di(C,D) \leq n-d$, therefore
there is at most one codeword with dimension greater than $n-d$.
Thus $\Ai(q,n,d) \leq 2$ for $d > \left\lfloor \frac{n}{2} \right\rfloor$ and $\Ai(q,n,d) \leq
\Ai(q,n,d,Q_d) + 2$ for $d \leq \left\lfloor \frac{n}{2} \right\rfloor$. Also, adding $\{ {\bf 0} \}$
and $\mathrm{GF}(q)^n$ to a code with minimum injection distance $d
\leq \left\lfloor \frac{n}{2} \right\rfloor$ and codewords of dimensions in $Q_d$ does not decrease the
minimum distance. Thus $\Ai(q,n,d) = \Ai(q,n,d,Q_d) + 2$ for $d \leq
\left\lfloor \frac{n}{2} \right\rfloor$. When $d > \left\lfloor \frac{n}{2} \right\rfloor$, $n-d \leq d$, and thus $\Am(q,n,d) = 2$.
\end{proof}

Proposition~\ref{prop:As<Am<As} below relates $\Ai(q,n,d)$ to
$\As(q,n,d)$ and $\Ac(q,n,r,d)$.

\begin{proposition}\label{prop:As<Am<As}
For all $q$, $n$, and $2 \leq d \leq \left\lfloor \frac{n}{2} \right\rfloor$, $\As(q,n,2d-1) \leq
\Ai(q,n,d) \leq \As(q,n,d)$; furthermore, when $d \geq \frac{n}{3}$,
$\Ai(q,n,d) \leq \As(q,n,4d-n,Q_d) + 2$. Also, $\Ac(q,n,\left\lfloor \frac{n}{2} \right\rfloor,d) \leq \Ai(q,n,d) \leq 2 + \sum_{r=d}^{n-d}
\Ac(q,n,r,d)$.
\end{proposition}

\begin{proof}
A code with minimum subspace distance $2d-1$ has minimum injection
distance $\geq d$ by (\ref{eq:ds_dm}) and hence $\As(q,n,2d-1) \leq
\Ai(q,n,d)$. Similarly, a code with minimum injection distance $d$
has minimum subspace distance $\geq d$ and hence $\Ai(q,n,d) \leq
\As(q,n,d)$.

Let $\mathcal{C}$ be a code with minimum injection distance $d$
whose codewords have dimensions in $Q_d$. For all codewords $U$ and
$V$, $\ds(U,V) = 2\di(U,V) - |\dim(U) - \dim(V)| \geq 2d - (n-2d)$.
Thus $\mathcal{C}$ has minimum subspace distance $4d-n \geq d$ for
$d \geq \frac{n}{3}$, and hence $\Ai(q,n,d,Q_d) \leq
\As(q,n,4d-n,Q_d)$. Proposition~\ref{prop:Am_2d>n} finally yields
$\Ai(q,n,d) = \Ai(q,n,d,Q_d) + 2 \leq \As(q,n,4d-n,Q_d) + 2$.

Any CDC in $E_r(q,n)$ with minimum injection distance $d$ is a
subspace code with minimum injection distance $d$, hence
$\Ac(q,n,r,d) \leq \Ai(q,n,d)$ for all $r$. Also, the codewords with
dimension $r$ in a subspace code with minimum injection distance $d$
form a CDC in $E_r(q,n)$ with minimum injection distance at least
$d$, hence $\Ai(q,n,d) = \Ai(q,n,d,Q_d) + 2 \leq 2 +
\sum_{r=d}^{n-d} \Ac(q,n,r,d)$.
\end{proof}



We now derive more bounds on $\Ai(q,n,d)$. Proposition \ref{prop:lower_Am} below is the analogue of Proposition \ref{prop:lower_As} for the injection metric, and its proof is hence omitted.

\begin{proposition} \label{prop:lower_Am}
For all $q$, $n$, and $2 \leq d \left\lfloor \frac{n}{2} \right\rfloor$, we have $\Ai(q,n,d) \geq 2 + \sum_{i=-z+1}^{z-1} \Ac(q,n,\left\lfloor \frac{n}{2} \right\rfloor - id,d)$, where $z = \left\lfloor \frac{\left\lfloor \frac{n}{2} \right\rfloor}{d} \right\rfloor$.
\end{proposition}

By extending the puncturing of subspaces introduced in
\cite{koetter_it08}, we finally derive below a Singleton bound for
injection metric codes.

\begin{proposition}\label{prop:singleton}
\textit{(Singleton bound for subspace codes in the injection metric).}
For all $q$, $n$, and $2 \leq d \leq \left\lfloor \frac{n}{2} \right\rfloor$, $\Ai(q,n,d) \leq
\Ai(q,n-1,d-1) \leq \sum_{r=0}^{n-d+1} {n-d+1 \brack r}$.
\end{proposition}

\begin{proof}
Let $W \in E_{n-1}(q,n)$. We define the puncturing $H_W(V)$ from
$E(q,n)$ to $E(q,n-1)$ as follows. If $\dim(V) = 0$, then
$\dim(H_W(V)) = 0$; otherwise, if $\dim(V) = r
> 0$, then $H_W(V)$ is a fixed $(r-1)$-subspace of $V \cap W$. For all
$U,V \in E(q,n)$, it is easily shown that $\di(H_W(U), H_W(V)) \geq
\di(U,V) - 1$, and hence $H_W(U) \neq H_W(V)$ if $\di(U,V) \geq 2$.

Therefore, if $\mathcal{C}$ is a code in $E(q,n)$ with minimum
injection distance $d \geq 2$, then $\{H_W(V) : V \in \mathcal{C}\}$
is a code in $E(q,n-1)$ with minimum injection distance $\geq d-1$
and cardinality $|\mathcal{C}|$. The first inequality follows.
Applying it $d-1$ times yields $\Ai(q,n,d) \leq \Ai(q,n-d+1,1) =
\sum_{r=0}^{n-d+1} {n-d+1 \brack r}$.
\end{proof}

We remark that although the puncturing defined in the proof of
Proposition \ref{prop:singleton} depends on $W$, the bounds in
Proposition \ref{prop:singleton} do not.

We now compare the cardinalities of optimal subspace codes and
optimal CDCs with the same minimum injection distance $d$. We first
establish the relation between $\Ai(q,n,d)$ and $\Ac(q,n,d)$ in
Proposition \ref{prop:Am_v_Ac} below.

\begin{proposition} \label{prop:Am_v_Ac}
\textit{(Comparison between optimal subspace codes and CDCs in the injection metric).}
For $2 \leq d \leq r \leq \left\lfloor \frac{n}{2} \right\rfloor$,
\begin{align*} \nonumber
    &q^{(\left\lfloor \frac{n}{2} \right\rfloor-r)(r-d+1)} \Ac \left(q,n,r, d \right)\\
    & \leq \Ai(q,n,d)\\
    &< 2 K_q^{-1} \theta(q) q^{(\left\lfloor \frac{n}{2} \right\rfloor-r)(r-d+1)} \Ac \left(q,n,r, d \right).
\end{align*}
\end{proposition}

The proof of Proposition~\ref{prop:Am_v_Ac} is similar to that of
Proposition~\ref{prop:As_v_Ac} and is hence omitted. We also obtain
another relation between $\Ai(q,n,d)$ and $\As(q,n,d)$.

\begin{corollary} \label{cor:Ai_v_As}
For $2 \leq d \leq \left\lfloor \frac{n}{2} \right\rfloor$, $\As(q,n,2d) \leq \As(q,n,2d-1) \leq
\Ai(q,n,d) < 2 K_q^{-1} \theta(q) \As(q,n,2d)$. Also, $\Ai(q,n,d) < 2
K_q^{-2} \theta(q) q^{(n-\left\lfloor \frac{n}{2} \right\rfloor)(\left\lfloor \frac{n}{2} \right\rfloor-d+1)}$.
\end{corollary}

\begin{proof}
The lower bounds on $\Ai(q,n,d)$ follow Proposition \ref{prop:As<Am<As}.
Furthermore,  by
choosing $r=\left\lfloor \frac{n}{2} \right\rfloor$ in Proposition~\ref{prop:Am_v_Ac} we have $\Ai(q,n,d) < 2 K_q^{-1} \theta(q) \Ac(q,n,\left\lfloor \frac{n}{2} \right\rfloor,d)$. Since
$\Ac(q,n,\left\lfloor \frac{n}{2} \right\rfloor,d) \leq \As(q,n,2d)$, we obtain $\Ai(q,n,d) < 2
K_q^{-1} \theta(q) \As(q,n,2d)$. The last inequality follows from
(\ref{eq:bounds_Ac}).
\end{proof}

Corollary~\ref{cor:Ai_v_As} provides several interesting insights.
First, the upper and lower bounds are all tight up to a scalar.
Second, for any optimal subspace code with minimum injection
distance $d$ and cardinality $\Ai(q,n,d)$, the optimal (or nearly
optimal) subspace codes with minimum subspace distance $2d$ have the
same cardinality up to a scalar. Third, the last inequality in
Corollary~\ref{cor:Ai_v_As} implies that such nearly optimal
subspace codes with minimum subspace distance $2d$ exist: KK codes
in $E_{\left\lfloor \frac{n}{2} \right\rfloor}(q,n)$ are such codes.

Based on Proposition~\ref{prop:Am_v_Ac}, we now determine the
asymptotic rates of subspace codes and CDCs with the injection
metric. Let us use the normalized parameters $r' = \frac{r}{n}$
defined earlier and $\di' \df \frac{\di}{n}$, where $\di$ is
the minimum injection distance of a code, and define the asymptotic
maximum rate $\ai(\di') = \lim \sup_{n \rightarrow \infty}
\frac{\log_q \Ai(q,n,r,\lceil n\di' \rceil)}{\log_q |E(q,n)|}$ for a subspace code
with the injection metric and the asymptotic rate $\ai(r',\di') =
\lim \sup_{n \rightarrow \infty} \frac{\log_q \Ac(q,n,nr',\lceil n\di' \rceil)}{\log_q
|E(q,n)|}$ for a CDC.

\begin{proposition}\label{prop:am}
\textit{(Asymptotic rate of packing subspace code in the injection metric).}
For $\frac{1}{2} \leq \di' \leq 1$, $\ai(\di') = 0$; or $0 \leq \di'
\leq \frac{1}{2}$, $\ai(\di') = 1-2\di'$. For $0 \leq r' \leq \di'$
or $1-\di' \leq r' \leq 1$, $\ai(r',\di') = 0$; for $\di' \leq r'
\leq \frac{1}{2}$, $\ai(r',\di') = 4(1-r')(r'-\di')$; for
$\frac{1}{2} \leq r' \leq 1-\di'$, $\ai(r',\di') = 4r'(1-r'-\di')$.
\end{proposition}

The proof of Proposition \ref{prop:am} is similar to that of
Proposition \ref{prop:as} and hence omitted.

Propositions~\ref{prop:Am_v_Ac} and \ref{prop:am} provide several
important insights on the design of subspace codes with the
injection metric. First, Proposition~\ref{prop:Am_v_Ac} indicates
that optimal CDCs with dimension being half of the block length up
to rounding ($r = \left\lfloor \frac{n}{2} \right\rfloor$ and $r =n- \left\lfloor \frac{n}{2} \right\rfloor$) are optimal subspace codes
with the injection metric up to a scalar. In this case, the optimal
CDCs have a limited rate loss as opposed to optimal subspace codes
with the same error correction capability.
When $r < \left\lfloor \frac{n}{2} \right\rfloor$, the rate loss suffered by optimal CDCs increases
with $\left\lfloor \frac{n}{2} \right\rfloor - r$. Proposition~\ref{prop:am} indicates that using CDCs
with dimension $\di \leq r < \left\lfloor \frac{n}{2} \right\rfloor$ leads to a decrease in rate on the
order of $(1-2r')(2\di'+1-2r')$. Similarly to the subspace metric,
the rate loss for CDCs using the injection metric increases with
$1-2r'$. Hence using a CDC with a dimension further from $\left\lfloor \frac{n}{2} \right\rfloor$ leads
to a high rate loss. The combinatorial explanation in
Section~\ref{sec:packing_subspace} also applies in this case.

We also comment that the asymptotic rates in
Proposition~\ref{prop:am} for subspace codes come from Singleton
bounds. The asymptotic rate $\ai(r',\di')$ is achieved by KK codes,
and the asymptotic rate $\ai(\di')$ is achievable also by KK codes
when $r = \left\lfloor \frac{n}{2} \right\rfloor$.

Proposition~\ref{prop:am} also compares the difference between
asymptotic rates of subspace codes with the subspace and injection
metrics. Although $\as(\ds')$ and $\ai(\di')$ are different, the
optimal subspace codes with the two metrics have similar asymptotic
behavior. We note that a CDC with minimum injection distance $\di$
has minimum subspace distance $\ds = 2\di$, which implies that $\as(r',\ds') = \ai(r',\di')$ as long as $\ds' = 2 \di'$. Also, as
shown above, CDCs in $E_{\left\lfloor \frac{n}{2} \right\rfloor}(q,n)$ with minimum injection distance
$\di$ are both asymptotically optimal subspace codes with minimum
subspace distance $\ds = 2\di$ and asymptotically optimal subspace
codes with minimum injection distance $\di$. Finally, when the asymptotic rate is fixed, the relative subspace distance $\ds'$ of optimal subspace codes is twice as much as the relative injection distance $\di'$. The implication of this on the error correction capability also depends on the decoding method.


In Table~\ref{table:comparison_injection}, we compare the bounds on $\Ai(q,n,d)$ derived in this paper with each other for $q=2$, $n=10$, and $d$ ranging from $2$ to $5$ (by Proposition \ref{prop:Am_2d>n}, $\Ai(2,10,d) = 2$ for $6 \leq d \leq 10$). We consider the lower bound in Proposition \ref{prop:As<Am<As} and its refinement in Proposition \ref{prop:lower_Am}, while the upper bound comes from Proposition~\ref{prop:As<Am<As}. Note that Proposition~\ref{prop:Am_v_Ac} is not included in the comparison since its primary purpose is to compare the cardinalities of optimal subspace codes and optimal CDCs with the same minimum injection distance. Although some bounds rely on $\Ac(q,n,r,d)$ whose values are unknown in general, the values in Table~\ref{table:comparison_injection} are obtained by using constructions  in \cite{etzion_it09} and \cite{gadouleau_it09_cdc} as lower bounds on $\Ac(q,n,r,d)$ and the upper bound  on $\Ac(q,n,r,d)$ in \cite{xia_dcc09}. 
The cardinalities of constant-dimension codes with dimension $r=n/2$ in \cite{etzion_it09} and \cite{gadouleau_it09_cdc} are quite close to the lower bound in Proposition \ref{prop:lower_Am}, again supporting our conclusion that the rate loss suffered by properly designed CDCs is smaller when the dimension is close to $n/2$. Finally, similar to the subspace distance case, the tightness of the bounds improves as the minimum distance of the code increases, leading to very tight bounds for $d=\left\lfloor \frac{n}{2} \right\rfloor$.

\begin{table}
\begin{center}
\begin{tabular}{|c|c|c||c|}
	\hline
	 & \multicolumn{2}{|c||}{Lower bounds} & Upper bound\\
	\hline
	$d$ & Proposition \ref{prop:As<Am<As} & Proposition \ref{prop:lower_Am} & Proposition \ref{prop:As<Am<As}\\
	\hline
	2 & 1,167,967 & 1,202,145 & 2,616,760\\
	3 & 32,843 & 32,843 & 50,708\\
	4 & 1,025 & 1,027 & 1,260\\
	5 & 33 & 35 & 35\\
	\hline
\end{tabular} \caption{Comparison of bounds on $\Ai(2,10,d)$ for $d$ from $2$ to $5$} \label{table:comparison_injection}
\end{center}
\end{table}

\subsection{Covering properties of subspace codes with the injection metric} \label{sec:covering_injection}

We now consider the covering properties of subspace codes with the
injection metric. The injection covering radius in $E(q,n)$ of $\mathcal{C}$ is defined as $\max_{U \in E(q,n)} \di(U,
\mathcal{C})$. We denote the minimum cardinality of a subspace code
with injection covering radius $\rho$ in $E(q,n)$ as
$\Ki(q,n,\rho)$. Since $\Ki(q,n,0) = |E(q,n)|$ and $\Ki(q,n,n) = 1$,
we assume $0 < \rho < n$ henceforth. We first
determine the minimum cardinality of a code with injection covering
radius $\rho$ when $\rho \geq \left\lfloor \frac{n}{2} \right\rfloor$.

\begin{proposition}\label{prop:Km_2rho>n}
For $n-\left\lfloor \frac{n}{2} \right\rfloor \leq \rho < n$, $\Ki(q,n,\rho) = 1$. If $n = 2\left\lfloor \frac{n}{2} \right\rfloor+1$,
then $\Ki(q,2\left\lfloor \frac{n}{2} \right\rfloor+1,\left\lfloor \frac{n}{2} \right\rfloor) = 2$.
\end{proposition}

\begin{proof}
Let $C$ be a subspace with dimension $\left\lfloor \frac{n}{2} \right\rfloor$. Then for all $D_1$ with
$\dim(D_1) \leq \dim(C)$, we have $\di(C,D_1) \leq \dim(C) = \left\lfloor \frac{n}{2} \right\rfloor$
by~(\ref{eq:dm1}); similarly, for all $D_2$ with $\dim(D_2) \geq
\dim(C) + 1$, we have $\di(C,D_2) \leq n - \dim(C) = n - \left\lfloor \frac{n}{2} \right\rfloor$
by~(\ref{eq:dm2}). Thus $C$ covers $E(q,n)$ with radius $n - \left\lfloor \frac{n}{2} \right\rfloor$
and $\Ki(q,n,\rho) = 1$ for $n - \left\lfloor \frac{n}{2} \right\rfloor \leq \rho < n$.

If $n = 2\left\lfloor \frac{n}{2} \right\rfloor+1$, then it is easily shown that $\{C, C^\perp\}$ has
covering radius $\left\lfloor \frac{n}{2} \right\rfloor$, and hence $\Ki(q,2\left\lfloor \frac{n}{2} \right\rfloor+1,\left\lfloor \frac{n}{2} \right\rfloor) \leq 2$.
However, for any $D \in E(q,2\left\lfloor \frac{n}{2} \right\rfloor+1)$, then either $\di(\{{\bf 0}\},
D) = \dim(D) > \left\lfloor \frac{n}{2} \right\rfloor$ or $\di(\mathrm{GF}(q)^n, D) = n-\dim(D) > \left\lfloor \frac{n}{2} \right\rfloor$.
Thus no single subspace can cover the projective space with radius $\left\lfloor \frac{n}{2} \right\rfloor$
and $\Ki(q,2\left\lfloor \frac{n}{2} \right\rfloor+1,\left\lfloor \frac{n}{2} \right\rfloor) \geq 2$.
\end{proof}

We thus consider $0 < \rho < \left\lfloor \frac{n}{2} \right\rfloor$ henceforth.
Lemma~\ref{lemma:Ks<Km<Ks} relates $\Ki(q,n,\rho)$ to
$\Ks(q,n,\rho)$ and $\Kc(q,n,r,\rho)$.

\begin{lemma}\label{lemma:Ks<Km<Ks}
For all $q$, $n$, and $0 < \rho < \left\lfloor \frac{n}{2} \right\rfloor$, $\Ks(q,n,2\rho) \leq
\Ki(q,n,\rho) \leq \Ks(q,n,\rho)$ and $\Ki(q,n,\rho) \leq 2 +
\sum_{r=\rho+1}^{n-\rho-1} \Kc(q,n,r,\rho)$.
\end{lemma}

\begin{proof}
A code with injection covering radius $\rho$ has subspace covering
radius $\leq 2\rho$, hence $\Ks(q,n,2\rho) \leq \Ki(q,n,\rho)$.
Also, a code with subspace covering radius $\rho$ has injection
covering radius $\leq \rho$, hence $\Ki(q,n,\rho) \leq
\Ks(q,n,\rho)$.

For $\rho + 1 \leq r \leq n-\rho-1$, let $\mathcal{C}_r$ be a CDC in
$E_r(q,n)$ with covering radius $\rho$ and cardinality
$\Kc(q,n,r,\rho)$ and let $\mathcal{C} =
\bigcup_{r=\rho+1}^{n-\rho-1} \mathcal{C}_r \cup \{ \{{\bf 0}\},
\mathrm{GF}(q)^n \}$. Then $\mathcal{C}$ is a subspace code with
injection covering radius $\rho$ and cardinality $2 +
\sum_{r=\rho+1}^{n-\rho-1} \Kc(q,n,r,\rho)$.
\end{proof}

Proposition~\ref{prop:sphere_bound_injection} below is the analogue
of Proposition~\ref{prop:sphere_bound_subspace} for the injection
metric.

\begin{proposition} \label{prop:sphere_bound_injection}
\textit{(Sphere covering bound for subspace codes in the injection metric).}
For all $q$, $n$, and $0 < \rho < \left\lfloor \frac{n}{2} \right\rfloor$, $\Ki(q,n,\rho) \geq \min
\sum_{i=0}^n A_i$, where the minimum is taken over all integer
sequences $\{A_i\}$ satisfying $0 \leq A_i \leq {n \brack i}$ for
all $0 \leq i \leq n$ and $\sum_{i=0}^n A_i \sum_{d=0}^\rho
\Ni(i,r,d) \geq {n \brack r}$ for $0 \leq r \leq n$.
\end{proposition}

The lower bound in Proposition \ref{prop:sphere_bound_injection} is again based on
the optimal solution to an integer linear program, and hence
determining the lower bound is computationally infeasible for large
parameter values.


Proposition \ref{prop:km'} below determines an upper bound on $\Ki(q,n,\rho)$, by applying the universal greedy algorithm in \cite[Theorem 12.2.1]{cohen_book97} to construct covering codes in the injection metric. Proposition \ref{prop:km'} is a direct application of the bound derived in \cite[Theorem 12.2.1]{cohen_book97} on the cardinality of a code returned by this algorithm. We remark that this bound is only semi-constructive, as it determines an algorithm to construct covering subspace codes but does not design the actual codes.

\begin{table*}
\begin{center}
\begin{tabular}{|c|c|c|c|}
	\hline
	& Properties & Subspace Metric & Injection Metric\\
	\hline
	\multirow{3}{*}{Packing} & asymptotic rates & $\as(\ds') = 1-\ds'$ & $\ai(\di') = 1-2\di'$\\
	\cline{2-4}
	& optimality of CDCs  with $r = n/2$ & optimal up to a scalar & optimal up to a scalar \\
	\cline{2-4}
	& optimal construction & optimal up to a scalar: KK codes & optimal up to a scalar: KK codes\\
	\hline
	\multirow{4}{*}{Covering} & asymptotic rates & $\ks(\rho') = 1-2\rho'$ & $\ki(\rho') = (1-2\rho')^2$\\
	\cline{2-4}
	& optimality of union of CDCs & not asymptotically optimal & asymptotically optimal\\
	\cline{2-4}
	& \multirow{2}{*}{optimal construction} & asymptotically optimal & asymptotically optimal \\
	& & semi-constructive bound: Prop.~\ref{prop:greedy_covering_subspace} & semi-constructive bound: Prop.~\ref{prop:km'}
\\
	\hline
\end{tabular} \caption{Summary of results} \label{table:results}
\end{center}
\end{table*}

\begin{proposition}\label{prop:km'}
\textit{(Greedy bound for covering codes in the injection metric).}
For all $q$, $n$, and $\rho$, $\Ki(q,n,\rho) \leq \frac{|E(q,n)|}{\min_{0 \leq r \leq n}
\Vi(r,\rho)}
    \left[ 1 + \ln\left(\max_{0 \leq r \leq n} \Vi(r,\rho)\right) \right]$.
\end{proposition}

We finally determine the asymptotic behavior of $\Ki(q,n,\rho)$ by using the asymptotic rate $\ki(\rho') = \lim \inf_{n \rightarrow \infty} \frac{\log_q \Ki(q,n,\lfloor n\rho' \rfloor)}{\log_q |E(q,n)|}$. According to Proposition \ref{prop:bound_Vm}, the volume of a ball with injection radius is constant up to a scalar. The consequence of this geometric result is that the greedy algorithm used to prove Proposition \ref{prop:km'} above will produce asymptotically optimal covering codes in the injection metric. However, since the volume of balls in the subspace metric does depend on the center (see Proposition \ref{prop:bound_Vs}), a direct application of the greedy algorithm for the subspace metric does not necessarily produce asymptotically optimal covering codes in the subspace metric.

\begin{proposition}\label{prop:km}
\textit{(Asymptotic rate of covering subspace code in the injection metric).}
For $0 \leq \rho' \leq \frac{1}{2}$, $\ki(\rho') = (1-2\rho')^2$.
For $\frac{1}{2} \leq \rho' \leq 1$, $\ki(\rho') = 0$.
\end{proposition}

\begin{proof}
By Proposition~\ref{prop:Km_2rho>n}, $\ki(\rho') = 0$ for
$\frac{1}{2} \leq \rho' \leq 1$. We have $\Ki(q,n,\rho) \geq
\frac{|E(q,n)|}{\max_{0 \leq r \leq n} \Vi(r,\rho)} >
\frac{K_q^{2}}{\theta(q)(2\theta(q)-1)} q^{\left\lfloor \frac{n}{2} \right\rfloor(n-\left\lfloor \frac{n}{2} \right\rfloor) - \rho(n-\rho)}$ by
Lemma~\ref{lemma:bounds_E} and Proposition~\ref{prop:bound_Vm}. This
asymptotically becomes $\ki(\rho') \geq (1-2\rho')^2$ for $0 \leq
\rho' \leq \frac{1}{2}$. Similarly, Proposition~\ref{prop:km'},
Lemma~\ref{lemma:bounds_E}, and Proposition~\ref{prop:bound_Vm}
yield 
\begin{align*}
	\Ki(q,n,\rho) <& 2K_q^{-1} \theta(q) q^{\left\lfloor \frac{n}{2} \right\rfloor(n-\left\lfloor \frac{n}{2} \right\rfloor) - \rho(n-\rho)}\\
	&\left[ 1 + \ln(\theta(q)(2\theta(q)-1) K_q^{-2}) + \rho(n-\rho) \ln q \right]
\end{align*}
which asymptotically becomes $\ki(\rho') \leq (1-2\rho')^2$ for $0
\leq \rho' \leq \frac{1}{2}$.
\end{proof}

The proof of Proposition \ref{prop:km} indicates that the minimum
cardinality $\Ki(q,n,\rho)$ of a covering subspace code with the
injection metric is on the order of $q^{\left\lfloor \frac{n}{2} \right\rfloor(n-\left\lfloor \frac{n}{2} \right\rfloor) - \rho(n-\rho)}$.
A covering subspace code is easily obtained by taking the union of
optimal covering CDCs for all constant dimensions, leading to a code
with cardinality $2 + \sum_{r=\rho+1}^{n-\rho-1} \Kc(q,n,r,\rho)$.
By \cite{gadouleau_it09_cdc}, the cardinality of the union is on the
order of $q^{\left\lfloor \frac{n}{2} \right\rfloor(n-\left\lfloor \frac{n}{2} \right\rfloor) - \rho(n-\rho)}$. Thus, a union of optimal
covering CDCs (in their respective Grassmannians) results in
asymptotically optimal covering subspace codes with the injection
metric.

Propositions~\ref{prop:ks} and \ref{prop:km} as well as their
implications illustrate the differences between the subspace and
injection metrics. First, the asymptotic rates of optimal covering
subspace codes with the two metrics are different. Second, a union of
optimal covering CDCs (in their respective Grassmannians) results in
asymptotically optimal covering subspace codes with the injection
metric only, not with the subspace metric. These differences can be
attributed to the different behaviors of the volume of a ball with
subspace and injection radius. Although $\Vs(0,t) = \Vi(0,t)$,
Proposition \ref{prop:bound_Vs} indicates that $\Vs(r,t)$ decreases
with $r$ ($r \leq \left\lfloor \frac{n}{2} \right\rfloor$), while according to Proposition
\ref{prop:bound_Vm}, $\Vi(r,t)$ remains asymptotically constant.
Hence, for $\left\lfloor \frac{n}{2} \right\rfloor-\rho \leq r \leq \left\lfloor \frac{n}{2} \right\rfloor$, the balls with subspace
radius $\rho$ centered at a subspace with dimension $r$ have
significantly smaller volumes than their counterparts with an injection
radius. Therefore, covering the subspaces with dimension $\left\lfloor \frac{n}{2} \right\rfloor$
requires more balls with subspace radius $\rho$ than balls with
injection radius $\rho$, which explains the different rates for
$\ks(\rho')$ and $\ki(\rho')$. Also, since the volume of a ball with
subspace radius reaches its minimum for $r = \left\lfloor \frac{n}{2} \right\rfloor$ and $E_{\left\lfloor \frac{n}{2} \right\rfloor}(q,n)$ has the largest cardinality among all Grassmannians, using covering
CDCs of dimension $\left\lfloor \frac{n}{2} \right\rfloor$ to cover $E_{\left\lfloor \frac{n}{2} \right\rfloor}(q,n)$ is not advantageous. Thus, a
union of covering CDCs does not lead to an asymptotically optimal
covering subspace code in the subspace metric.


\section{Conclusion} \label{sec:conclusion}

In this paper, we derive packing and covering properties of subspace codes for the subspace and the injection metrics. We determine the asymptotic rates of packing and covering codes for both metrics, compare the performance of constant-dimension codes to that of general subspace codes, and provide constructions or semi-constructive bounds of nearly optimal codes in all four cases. These results are briefly summarized in Table \ref{table:results}.

Despite these results, some open problems remain for subspace codes. First of all, our bounds on the volumes of balls derived in Lemma \ref{lemma:bounds_E} and Propositions \ref{prop:bound_Vs} and \ref{prop:bound_Vm} may be tightened. Although the ratio between the upper and lower bounds is a function of the field size $q$ which tends to $1$ as $q$ tends to infinity, it is unknown whether this ratio is the smallest that can be established. This issue also applies to the bounds on packing subspace codes in Propositions \ref{prop:As_v_Ac} and \ref{prop:Am_v_Ac}, where the ratios between upper and lower bounds are similar functions of $q$. Also, we only considered balls with radii up to $\frac{n}{2}$, as only this case was useful for our derivations; the case where the radius is above $\frac{n}{2}$ remains unexplored. Second, the bounds on covering codes in both the subspace and the injection metrics derived in this paper are only asymptotically optimal. It remains unknown whether any of these bounds is tight up to a scalar. Third, the design of packing and covering subspace codes is an important topic for future work. This is especially the case for covering codes in the subspace metric, as no asymptotically optimal construction is known so far. Finally, the aim of this paper was to derive simple bounds on subspace codes which are good for all parameter values, especially large values. On the other hand, a wealth of ad hoc bounds and heuristics can be used to tighten our results for small parameter values.

\section{Acknowledgment}
The authors are grateful to the anonymous reviewers and the
associate editor Dr.~Mario Blaum for their constructive comments,
which have helped to improve this paper. 

\appendix
\subsection{Proof of Proposition~\ref{prop:bound_Vs}}
\label{app:prop:bound_Vs}

\begin{proof}
When $r=0$, we have $g(0,t) = t(n-t)$ and $\Vs(0,t) = \sum_{i=0}^t {n \brack i}$ for all $t \leq \frac{n}{2}$. Hence $\Vs(0,t)
\geq {n \brack t} \geq q^{t(n-t)}$ by (\ref{eq:Gaussian}), which
proves the lower bound. Also, $\Vs(0,t) < K_q^{-1} \sum_{i=0}^t
q^{i(n-i)} < K_q^{-1} q^{t(n-t)} \sum_{j=0}^\infty q^{-j^2}$ by
(\ref{eq:Gaussian}), which proves the upper bound.

We now prove the bounds on $\Vs(r,t)$ for $r \geq 1$. By definition, $\Vs(r,t) = \sum_{s=0}^n \sum_{d=0}^t \Ns(r,s,d)$ is a double summation of exponential terms. The main idea of the proof is to determine the largest term in the summation: this not only gives a good lower bound, but the whole summation can also be upper bounded by that term times a constant. First, by
Lemma~\ref{lemma:Ns}, $\Ns(r,s,d) = q^{u(d-u)} {r \brack u} {n-r
\brack d-u}$, where $u = \frac{r+d-s}{2}$ satisfies $0 \leq u \leq
\min\{r,d\}$. Thus $q^{f(u)} \leq \Ns(r,s,d) < K_q^{-2} q^{f(u)}$ by
(\ref{eq:Gaussian}), where $f(u) = u(2r+3d-n-3u) + d(n-r-d)$. Hence,
$\sum_{d=0}^t S(d) \leq \Vs(r,t) < K_q^{-2} \sum_{d=0}^t S(d)$,
where $S(d) = \sum_{u=0}^{\min\{r,d\}} q^{f(u)}$. Since $f$ is
maximized for $u = u_0 \df \frac{2r+3d-n}{6} \leq d$, we need to
consider the following three cases.

\begin{itemize}
    \item Case I: $0 \leq d \leq \frac{n-2r}{3}$.
    We have $u_0 \leq 0$ and hence $f$ is maximized for
    $u=0$: $f(0) = g(r,d) = d(n-r-d)$. Thus $S(d) \geq
    q^{g(r,d)}$, and it is easy to show that
    $S(d) = q^{g(r,d)} \sum_{u=0}^{\min\{r,d\}} q^{-u(n-2r-3d + 3u)}
    < \theta(q^3) q^{g(r,d)}$ since $n-2r-3d \geq 0$.

    \item Case II: $\frac{n-2r}{3} \leq d \leq
    \min\left\{\frac{n+4r}{3}, \frac{n}{2}\right\}$.
    We have $0 \leq u_0 \leq r$ and hence $f$ is maximized for $u = u_0$:
    $f(u_0)= g(r,d) = \frac{1}{12}(n-2r)^2 +
    \frac{1}{4}d(2n-d)$. It is easily shown that $f(u) = f(u_0) - 3(u-u_0)^2$ for all $u$ and hence $S(d) \geq
    \max\{q^{f(\lfloor u_0 \rfloor)}, q^{f(\lceil u_0 \rceil)}\}
    \geq q^{g(r,d) - \frac{3}{4}}$. We also obtain $S(d) = q^{g(r,d)}
    \sum_{u=0}^{\min\{r,d\}} q^{-3(u-u_0)^2} < 2\theta(q^3) q^{g(r,d)}$.

    \item Case III: $\frac{n+4r}{3} \leq d \leq
    \frac{n}{2}$. We have $u_0 \geq r$ and hence $f$ is maximized for $u=r$:
    $f(r) = g(r,d) = (d-r)(n-d+r)$.
    Thus $S(d) \geq q^{g(r,d)}$, and it is easy to
    show that $S(d) = q^{g(r,d)} \sum_{i=0}^r q^{-i(3d-4r-n+3i)}
    < \theta(q^3) q^{g(r,d)}$ since $3d-4r-n \geq 0$.
\end{itemize}

From the discussion above, we obtain $\Vs(r,t) \geq S(t) \geq
q^{-\frac{3}{4} + g(r,t)}$ which proves the lower bound, and
$\Vs(r,t) < K_q^{-2} \sum_{d=0}^t S(d) < 2 \theta(q^3) K_q^{-2} \sum_{d=0}^t
q^{g(r,d)}$. We now show that $R(t) = \sum_{d=0}^t q^{g(r,d)} < (1 +
q^{-1}) \theta(q^{\frac{3}{4}}) q^{g(r,t)}$ by distinguishing the following three cases.

First, if $t \leq \frac{n-2r}{3}$, $R(t) = \sum_{d=0}^t q^{d(n-r-d)}
= q^{t(n-r-t)} \sum_{i=0}^t q^{-i(n-r-2t+i)} < q^{g(r,t)}
\theta(q)$ since $n-2r-2t \geq 0$.

Second, if $\frac{n-2r}{3} < t \leq \frac{n+4r}{3}$, we have $(n-2r-3d)^2 = \frac{1}{12}(n-2r)^2 + \frac{1}{4}d(2n-d) - d(n-r-d)$ and hence
$\frac{1}{12}(n-2r)^2 + \frac{1}{4}d(2n-d) \geq d(n-r-d)$ for all
$d$. We obtain $R(t) = \sum_{d=0}^{\left\lfloor \frac{n-2r}{3} \right\rfloor} q^{d(n-r-d)} + \sum_{d=\left\lfloor \frac{n-2r}{3} \right\rfloor + 1}^t q^{\frac{1}{12}(n-2r)^2 + \frac{1}{4}d(2n-d)} \leq \sum_{d=0}^t q^{\frac{1}{12}(n-2r)^2 +
\frac{1}{4}d(2n-d)}$ and hence $R(t) = q^{g(r,t)} \sum_{i=0}^t
q^{-\frac{1}{4}i(2n-2t+i)} < \theta(q^{\frac{3}{4}}) q^{g(r,t)}$ since $2n-2t \geq 2t$.

Third, if $\frac{n+4r}{3} < t \leq \frac{n}{2}$, which implies $1
\leq r < \frac{n}{8}$, it can be shown that $g(r,\left\lfloor
\frac{n+4r}{3} \right\rfloor) \leq g(r,\left\lfloor \frac{n+4r}{3}
\right\rfloor + 1) - \frac{n-2r}{3} + 1 \leq g(r,t) - \frac{4}{3}$. Hence
$R(t) = R\left(\left\lfloor \frac{n+4r}{3} \right\rfloor\right) +
q^{g(r,t)} \sum_{j=0}^{t - \left\lfloor \frac{n+4r}{3}
\right\rfloor-1}
    q^{-j(n-2t+2r+j)} < q^{g(r,t)-\frac{4}{3}} \theta(q^{\frac{3}{4}}) + q^{g(r,t)} \theta(q)$.

Thus, $\Vs(r,t) < 2\theta(q^3) (1 + q^{-\frac{4}{3}})\theta(q^{\frac{3}{4}}) K_q^{-2} q^{g(r,t)}$.
\end{proof}

\subsection{Proof of
Proposition~\ref{prop:bound_Vm}}\label{app:prop:bound_Vm}

\begin{proof}
First, $\Vi(r,t) \geq \Ni(r,r,t) \geq q^{t(n-t)}$. We now prove the upper bound by determining the largest term in the double summation of $\Vi(r,t)$. Since $\Vi(r,t) = \Vi(n-r,t)$, we assume $r \leq \left\lfloor \frac{n}{2} \right\rfloor$ without loss of generality. The triangular inequality indicates that $\Ni(r,s,d) = 0$ if $s > |r-d|$ or $s > r+d$; also, by definition of the injection distance, $\Ni(r,s,d) = 0$ if $d > \max\{r,s\}$. We can hence restrict the range of parameters in the summation formula of $\Vi(r,t)$ as follows:
\begin{equation} \label{eq:range_Vi}
	\Vi(r,t) = \sum_{d=0}^r \sum_{s=r-d}^{d+r} \Ni(r,s,d) + \sum_{d=r+1}^t \sum_{s=d}^{d+r} \Ni(r,s,d).
\end{equation}
By Lemma~\ref{lemma:Nm} and \ref{eq:Gaussian}, we have $\Ni(r,s,d) < K_q^{-2} q^{s(n-d+r-s) - (r-d)(n-d)}$ for $s \leq r$ and $\Ni(r,s,d) < K_q^{-2} q^{s(r-s+d) + d(n-r-d)}$ for $s \geq r$, which with (\ref{eq:range_Vi}) yields
\begin{align}
    \nonumber
    &K_q^{2} \Vi(r,t)\\
    \nonumber
    &< \sum_{d=0}^r \left\{ \sum_{s=r-d}^r q^{s(n-d+r-s) - (r-d)(n-d)}\right.\\
    \nonumber
    &+ \left. \sum_{s=r+1}^{r+d} q^{s(r-s+d) + d(n-r-d)} \right\}\\
    \nonumber
    &+ \sum_{d=r+1}^t \sum_{s=d}^{d+r} q^{s(r-s+d) + d(n-r-d)}\\
    \nonumber
    &= \sum_{d=0}^r  q^{d(n-d)} \left\{ \sum_{i=0}^d q^{-i(n-d-r+i)}
    + \sum_{j=1}^d q^{-j(r-d+j)} \right\}\\
    \label{eq:Vi1}
    &+ \sum_{d=r+1}^t q^{d(n-d)} \sum_{k=0}^r q^{-k(d-r+k)},
\end{align}
where we make the following changes of variables: $i=r-s$, $j=s-r$, $k=s-d$ in (\ref{eq:Vi1}). Since $n-d-r \geq 0$, we have $\sum_{i=0}^d q^{-i(n-d-r+i)} < \theta(q)$. Also, $r-d \geq 0$ for $r \geq d$, and hence $\sum_{j=1}^d q^{-j(r-d+j)} < \theta(q) - 1$; similarly, we obtain $\sum_{k=0}^r q^{-k(d-r+k)} < \theta(q)$. Hence, (\ref{eq:Vi1}) leads to
\begin{eqnarray}
    \label{eq:Vi2} \nonumber
    K_q^{2} \Vi(r,t) &<& (2 \theta(q) - 1) \sum_{d=0}^r q^{d(n-d)} + \theta(q) \sum_{d=r+1}^t q^{d(n-d)}\\
    \label{eq:Vi3} 
    &<& (2\theta(q) - 1) q^{t(n-t)} \sum_{l=0}^t q^{-l(n-2t+l)}\\
    \nonumber
    &<& (2 \theta(q) - 1) \theta(q) q^{t(n-t)},
\end{eqnarray}
where we set $l=t-d$ and use $n\geq 2r$ in (\ref{eq:Vi3}).
\end{proof}

\bibliographystyle {IEEEtr}

\end{document}